\documentclass[a4paper,11pt]{article}
\pdfoutput=1
\usepackage[utf8]{inputenc}
\usepackage{amsmath}
\usepackage{amsfonts}
\usepackage{amssymb}
\usepackage{amsthm}
\usepackage{environ}
\usepackage{xcolor,graphicx}
\usepackage{verbatim}
\usepackage{a4wide}

\usepackage{subfig}
\usepackage{paralist}
\usepackage{xspace}
\usepackage{tabularx}
\usepackage{multirow}
\usepackage{graphicx}
\usepackage{enumerate}
\usepackage{mathtools}
\usepackage{dsfont}
\usepackage[numbers,sort]{natbib}
\usepackage{booktabs}

\usepackage[vlined,linesnumbered,ruled]{algorithm2e}
\SetEndCharOfAlgoLine{}
\SetCommentSty{textrm}

\usepackage{authblk}

\usepackage{scrpage2}

\usepackage{tikz}
\tikzset{%
  every neuron/.style={
    circle,
    draw,
    minimum size=0.8cm
  },
  neuron missing/.style={
    draw=none, 
    scale=3,
    text height=0.333cm,
    execute at begin node=\color{black}$\vdots$
  },
}
\usetikzlibrary{positioning}
\usetikzlibrary{fit}
\usetikzlibrary{shapes,arrows}
\usetikzlibrary{matrix}
\usetikzlibrary{decorations.pathreplacing, decorations.markings, arrows}

\usetikzlibrary{arrows,positioning} 

\tikzset{
	>=stealth',
	dot/.style={
		radius=0.05,
		fill=black,
	},
	gnode/.style={
          draw,
          shape=circle,
          inner sep=4pt
	},
	dummynode/.style={
		  gnode
		, fill=black
	},
	edgeto/.style={
		->,
		shorten <=2pt,
		shorten >=2pt,
		bend right=15,},
	edgeboth/.style={
		<->,
		shorten <=2pt,
		shorten >=2pt,
		},
	edgefrom/.style={
		<-,
		shorten <=2pt,
		shorten >=2pt,
		bend left=15,}
}

\usepackage[pagebackref]{hyperref}

\usepackage{todonotes}

\theoremstyle{remark}

\theoremstyle{plain}
\newtheorem{theorem}{Theorem}
\newtheorem{lemma}[theorem]{Lemma}
\newtheorem{corollary}[theorem]{Corollary}

\newtheorem{construction}[theorem]{Construction}
\newtheorem{rrule}{Reduction Rule}[section]

\theoremstyle{definition}
\newtheorem{definition}[theorem]{Definition}

\usepackage[sort&compress]{cleveref}
\crefname{ineq}{inequality}{inequalities}
\creflabelformat{ineq}{#2{\upshape(#1)}#3} 

\crefname{cond}{Condition}{Conditions}
\creflabelformat{cond}{#2(#1)#3}

\crefname{algorithm}{Algorithm}{Algorithms}
\crefname{line}{Line}{Lines}
\crefname{rrule}{Reduction Rule}{Reduction Rules}
\crefname{chapter}{Chapter}{Chapters}
\crefname{section}{Section}{Sections}
\crefname{subsection}{Section}{Sections}
\crefname{theorem}{Theorem}{Theorems}
\crefname{obs}{Observation}{Observations}
\crefname{proposition}{Proposition}{Propositions}
\crefname{corollary}{Corollary}{Corollaries}
\crefname{lemma}{Lemma}{Lemmas}
\crefname{figure}{Figure}{Figures}
\crefname{construction}{Construction}{Constructions}
\crefname{definition}{Definition}{Definitions}

\newcommand{\NCEA}{\textsc{NCE}\xspace}

\newcommand{\numO}{\textsc{\#}}
\newcommand{\numOlong}{\textsc{Numbers Only}\xspace}

\newcommand{\DDCSClong}{\textsc{Digraph Degree Constraint Sequence Completion}\xspace}
\newcommand{\DDCSC}{\textsc{DDConSeqC}\xspace}
\newcommand{\nDDCSClong}{\numOlong \DDCSClong} 
\newcommand{\nDDCSC}{\numO\DDCSC}
\newcommand{\DDCClong}{\textsc{Digraph Degree Constraint Completion}\xspace}
\newcommand{\DDCC}{\textsc{DDConC}\xspace}
\newcommand{\nDDCClong}{\numOlong \DDCClong}
\newcommand{\nDDCC}{\numO\DDCC}

\newcommand{\DDSClong}{\textsc{Digraph Degree Sequence Completion}\xspace}
\newcommand{\DDSC}{\textsc{DDSeqC}\xspace}
\newcommand{\nDDSClong}{\numOlong \DDSClong} %
\newcommand{\nDDSC}{\numO\DDSC} %
\newcommand{\DAlong}{\textsc{Digraph Degree Anonymity}\xspace}
\newcommand{\DA}{\textsc{DDA}\xspace}
\newcommand{\nDAlong}{\numOlong \DAlong} %
\newcommand{\nDA}{\numO\DA} %

\DeclareMathOperator{\diff}{diff}

\newcommand{\N}{\ensuremath{\mathds{N}}}
\newcommand{\degpara}{\ensuremath{\Delta^*}\xspace}

\makeatletter
\newcommand{\problemtitle}[1]{\gdef\@problemtitle{#1}}% Store problem title
\newcommand{\probleminput}[1]{\gdef\@probleminput{#1}}% Store problem input
\newcommand{\problemquestion}[1]{\gdef\@problemquestion{#1}}% Store problem question
\NewEnviron{problemdef}{
  \problemtitle{}\probleminput{}\problemquestion{}% Default input is empty
  \BODY% Parse input
  \par\addvspace{.5\baselineskip}
  \noindent
  \begin{tabularx}{\textwidth}{@{\hspace{0ex}} l X c}
    \multicolumn{2}{@{\hspace{0ex}}l}{\normalsize\@problemtitle} \\% Title
    \normalsize\textbf{Input:} & \normalsize\@probleminput \\% Input
    \normalsize\textbf{Question:} & \normalsize\@problemquestion% Question
  \end{tabularx}
  \par\addvspace{.5\baselineskip}
}

\usepackage{etoolbox}

\begin{document}

\title{A Parameterized Algorithmics Framework for Degree Sequence Completion Problems in Directed Graphs}

\author[]{Robert Bredereck}
\author[]{Vincent Froese}
\author[]{Marcel Koseler}
\author[]{Marcelo Garlet Millani}
\author[]{André~Nichterlein}
\author[]{Rolf Niedermeier}

\affil[]{Institut f\"ur Softwaretechnik und Theoretische Informatik,  TU Berlin, Germany,
 \texttt{\{robert.bredereck, vincent.froese, m.garletmillani, andre.nichterlein, rolf.niedermeier\}@tu-berlin.de}}
\date{}

\maketitle

\thispagestyle{scrheadings}
\cfoot{}
\ohead{}
\ifoot{}

\begin{abstract}
There has been intensive work on the parameterized complexity of
the typically NP-hard task to 
edit undirected graphs into graphs fulfilling certain given vertex degree 
constraints. In this work, we lift the investigations to the case of directed graphs; herein, we focus on arc insertions.
To this end, we develop a general two-stage framework which consists of efficiently solving a problem-specific number problem and transferring its solution to a solution for the graph problem by applying flow computations. In this way, we obtain fixed-parameter tractability and polynomial kernelizability
results, with the central parameter being the maximum vertex in- or outdegree of the output digraph.
Although there are certain similarities with the much better studied undirected case,
the flow computation used in the directed case seems not to 
work for the undirected case while $f$-factor computations as used 
in the undirected case seem not to work for the directed case.
\end{abstract}

\section{Introduction}
Modeling real-world networks (e.g., communication, ecological, social) 
often requests \emph{directed} graphs (digraphs for short). 
We study a class of specific ``network design'' (in the sense of constructing
a specific network topology) or ``graph realization''
problems. Here, our focus is on inserting arcs into a given digraph in order to fulfill certain vertex degree constraints. These problems are 
typically NP-hard, so we choose parameterized algorithm design 
for identifying relevant tractable special cases. The main parameter 
we work with is the maximum in- or outdegree of the newly constructed digraph.
We deal with the following three problems:
First, the problem \DDCClong (\DDCC), which asks to insert a minimum number of arcs
such that each vertex ends up with a degree as specified by an individual list of target degrees (see \Cref{fig:example DDCC}).
Second, the \DDSClong (\DDSC) problem, where the goal is to insert arcs in such a way that
the resulting digraph has a specific degree sequence (\Cref{fig:example DDSC}).
Third, \DAlong (\DA), which asks to ``$k$-anonymize'' a given digraph, that is, after inserting a minimum number of arcs, each combination of in- and outdegree occurs either zero or at least $k$~times (\Cref{fig:example DA}).
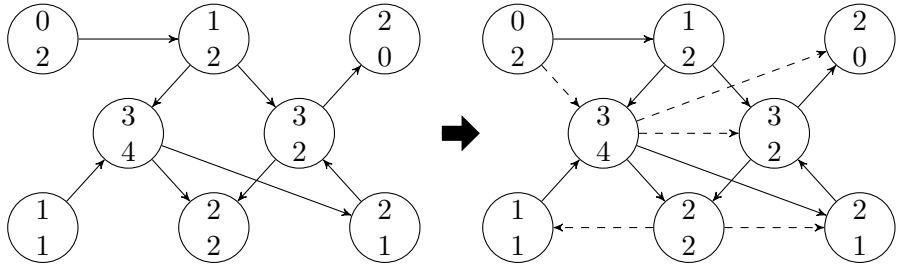
\begin{figure}[t]
	\centering
	\begin{tikzpicture}
	\node[gnode, inner sep=2.0pt, radius=3pt, minimum size=16pt,  draw=black] (path4136-5) at (0.0, 0.0) [align=left]{$1$\\$1$};
\node[gnode, inner sep=2.0pt, radius=3pt, minimum size=16pt,  draw=black] (path4136-62) at (2.25, 0.0) [align=left]{$2$\\$2$};
\node[gnode, inner sep=2.0pt, radius=3pt, minimum size=16pt,  draw=black] (path4136-9) at (4.5, 0.0) [align=left]{$2$\\$1$};
\node[gnode, inner sep=2.0pt, radius=3pt, minimum size=16pt,  draw=black] (path4136-7) at (1.125, 1.25) [align=left]{$3$\\$4$};
\node[gnode, inner sep=2.0pt, radius=3pt, minimum size=16pt,  draw=black] (path4136-35) at (3.375, 1.25) [align=left]{$3$\\$2$};
\node[gnode, inner sep=2.0pt, radius=3pt, minimum size=16pt,  draw=black] (path4136) at (0.0, 2.5) [align=left]{$0$\\$2$};
\node[gnode, inner sep=2.0pt, radius=3pt, minimum size=16pt,  draw=black] (path4136-3) at (2.25, 2.5) [align=left]{$1$\\$2$};
\node[gnode, inner sep=2.0pt, radius=3pt, minimum size=16pt,  draw=black] (path4136-6) at (4.5, 2.5) [align=left]{$2$\\$0$};
\draw[draw=black, ->] (path4136-5) edge (path4136-7);
\draw[draw=black, ->] (path4136-3) edge (path4136-35);
\draw[draw=black, ->] (path4136-35) edge (path4136-6);
\draw[draw=black, ->] (path4136-9) edge (path4136-35);
\draw[draw=black, ->] (path4136-35) edge (path4136-62);
\draw[draw=black, ->] (path4136-7) edge (path4136-62);
\draw[draw=black, ->] (path4136-7) edge (path4136-9);
\draw[draw=black, ->] (path4136) edge (path4136-3);
\draw[draw=black, ->] (path4136-3) edge (path4136-7);

\draw[
        -triangle 90,
        line width=0.5mm,
        postaction={draw, line width=2mm, shorten >=2mm, -}
    ] (5.25,1.25) -- (5.75,1.25);

	\begin{scope}[shift={(6.25,0)}]
	\node[gnode, inner sep=2.0pt, radius=3pt, minimum size=16pt,  draw=black] (path4136-5) at (0.0, 0.0) [align=left]{$1$\\$1$};
\node[gnode, inner sep=2.0pt, radius=3pt, minimum size=16pt,  draw=black] (path4136-62) at (2.25, 0.0) [align=left]{$2$\\$2$};
\node[gnode, inner sep=2.0pt, radius=3pt, minimum size=16pt,  draw=black] (path4136-9) at (4.5, 0.0) [align=left]{$2$\\$1$};
\node[gnode, inner sep=2.0pt, radius=3pt, minimum size=16pt,  draw=black] (path4136-7) at (1.125, 1.25) [align=left]{$3$\\$4$};
\node[gnode, inner sep=2.0pt, radius=3pt, minimum size=16pt,  draw=black] (path4136-35) at (3.375, 1.25) [align=left]{$3$\\$2$};
\node[gnode, inner sep=2.0pt, radius=3pt, minimum size=16pt,  draw=black] (path4136) at (0.0, 2.5) [align=left]{$0$\\$2$};
\node[gnode, inner sep=2.0pt, radius=3pt, minimum size=16pt,  draw=black] (path4136-3) at (2.25, 2.5) [align=left]{$1$\\$2$};
\node[gnode, inner sep=2.0pt, radius=3pt, minimum size=16pt,  draw=black] (path4136-6) at (4.5, 2.5) [align=left]{$2$\\$0$};
\draw[draw=black, ->] (path4136-5) edge (path4136-7);
\draw[draw=black, dashed, ->] (path4136) edge (path4136-7);
\draw[draw=black, ->] (path4136-3) edge (path4136-35);
\draw[draw=black, ->] (path4136-35) edge (path4136-6);
\draw[draw=black, dashed, ->] (path4136-7) edge (path4136-35);
\draw[draw=black, dashed, ->] (path4136-62) edge (path4136-9);
\draw[draw=black, ->] (path4136-9) edge (path4136-35);
\draw[draw=black, ->] (path4136-35) edge (path4136-62);
\draw[draw=black, ->] (path4136-7) edge (path4136-62);
\draw[draw=black, dashed, ->] (path4136-7) edge (path4136-6);
\draw[draw=black, ->] (path4136-7) edge (path4136-9);
\draw[draw=black, ->] (path4136) edge (path4136-3);
\draw[draw=black, ->] (path4136-3) edge (path4136-7);
\draw[draw=black, dashed, ->] (path4136-62) edge (path4136-5);
\end{scope}
	\end{tikzpicture}
	\caption{Example instance of \DDCClong.
	The numbers inside the vertices represent the desired degrees (indegree on top, outdegree on bottom).
	We can satisfy all demands by adding the dashed arcs.
	Note that, in general, the desired degrees of a vertex do not have to be unique.}
	\label{fig:example DDCC}
\end{figure}

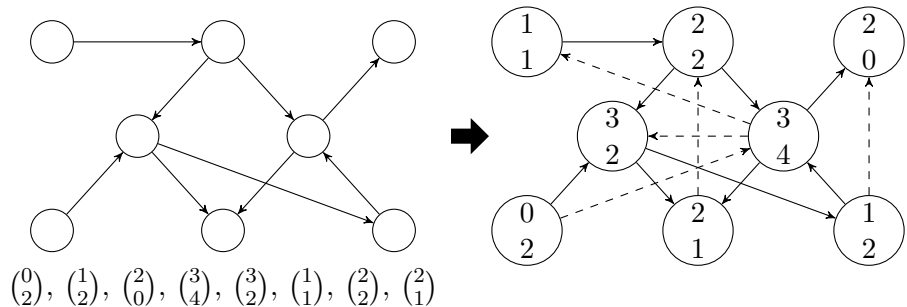
\begin{figure}[t]
	\centering
	\begin{tikzpicture}[xscale=1.5]
		
		\node[gnode, inner sep=2.0pt, radius=3pt, minimum size=16pt, fill=white, draw=black] (path4136-5) at (0.0, 0.0) [align=left]{};
		\node[gnode, inner sep=2.0pt, radius=3pt, minimum size=16pt, fill=white, draw=black] (path4136-62) at (1.5, 0.0) [align=left]{};
		\node[gnode, inner sep=2.0pt, radius=3pt, minimum size=16pt, fill=white, draw=black] (path4136-9) at (3.0, 0.0) [align=left]{};
		\node[gnode, inner sep=2.0pt, radius=3pt, minimum size=16pt, fill=white, draw=black] (path4136-7) at (0.75, 1.25) [align=left]{};
		\node[gnode, inner sep=2.0pt, radius=3pt, minimum size=16pt, fill=white, draw=black] (path4136-35) at (2.25, 1.25) [align=left]{};
		\node[gnode, inner sep=2.0pt, radius=3pt, minimum size=16pt, draw=black] (path4136) at (0.0, 2.5) [align=left]{};
		\node[gnode, inner sep=2.0pt, radius=3pt, minimum size=16pt, draw=black] (path4136-3) at (1.5, 2.5) [align=left]{};
		\node[gnode, inner sep=2.0pt, radius=3pt, minimum size=16pt, draw=black] (path4136-6) at (3.0, 2.5) [align=left]{};
		\node[rectangle, thick, dashed] (rect5360) at (1.5, -0.75) [align=left]{\normalsize$0 \choose 2$, $1 \choose 2$, $2 \choose 0$, $3 \choose 4$, $3 \choose 2$, $1 \choose 1$, $2 \choose 2$, $2 \choose 1$};
		\draw[draw=black, ->] (path4136-5) edge (path4136-7);
		\draw[draw=black, ->] (path4136-3) edge (path4136-35);
		\draw[draw=black, ->] (path4136-35) edge (path4136-6);
		\draw[draw=black, ->] (path4136-9) edge (path4136-35);
		\draw[draw=black, ->] (path4136-35) edge (path4136-62);
		\draw[draw=black, ->] (path4136-7) edge (path4136-62);
		\draw[draw=black, ->] (path4136-7) edge (path4136-9);
		\draw[draw=black, ->] (path4136) edge (path4136-3);
		\draw[draw=black, ->] (path4136-3) edge (path4136-7);
		
		\draw[
        -triangle 90,
        line width=0.5mm,
        postaction={draw, line width=2mm, shorten >=2mm, -}
    ] (3.5,1.25) -- (3.83333,1.25);
		
		\begin{scope}[shift={(4.16666,0)}]
			\node[gnode, inner sep=2.0pt, radius=3pt, minimum size=16pt, fill=white, draw=black] (path4136-5) at (0.0, 0.0) [align=left]{$0$\\$2$};
		\node[gnode, inner sep=2.0pt, radius=3pt, minimum size=16pt, fill=white, draw=black] (path4136-62) at (1.5, 0.0) [align=left]{$2$\\$1$};
		\node[gnode, inner sep=2.0pt, radius=3pt, minimum size=16pt, fill=white, draw=black] (path4136-9) at (3.0, 0.0) [align=left]{$1$\\$2$};
		\node[gnode, inner sep=2.0pt, radius=3pt, minimum size=16pt, fill=white, draw=black] (path4136-7) at (0.75, 1.25) [align=left]{$3$\\$2$};
		\node[gnode, inner sep=2.0pt, radius=3pt, minimum size=16pt, fill=white, draw=black] (path4136-35) at (2.25, 1.25) [align=left]{$3$\\$4$};
		\node[gnode, inner sep=2.0pt, radius=3pt, minimum size=16pt, draw=black] (path4136) at (0.0, 2.5) [align=left]{$1$\\$1$};
		\node[gnode, inner sep=2.0pt, radius=3pt, minimum size=16pt, draw=black] (path4136-3) at (1.5, 2.5) [align=left]{$2$\\$2$};
		\node[gnode, inner sep=2.0pt, radius=3pt, minimum size=16pt, draw=black] (path4136-6) at (3.0, 2.5) [align=left]{$2$\\$0$};
		%\node[rectangle, thick, dashed] (rect5360) at (1.5, -0.75) [align=left]{\normalsize$0 \choose 2$, $1 \choose 2$, $2 \choose 0$, $3 \choose 4$, $3 \choose 2$, $1 \choose 1$, $2 \choose 2$, $2 \choose 1$};
		\draw[draw=black, ->] (path4136-5) edge (path4136-7);
		\draw[draw=black, ->] (path4136-3) edge (path4136-35);
		\draw[draw=black, ->] (path4136-35) edge (path4136-6);
		\draw[draw=black, ->] (path4136-9) edge (path4136-35);
		\draw[draw=black, ->] (path4136-35) edge (path4136-62);
		\draw[draw=black, ->] (path4136-7) edge (path4136-62);
		\draw[draw=black, ->] (path4136-7) edge (path4136-9);
		\draw[draw=black, ->] (path4136) edge (path4136-3);
		\draw[draw=black, ->] (path4136-3) edge (path4136-7);
		\draw[draw=black, dashed, ->] (path4136-5) edge (path4136-35);
		\draw[draw=black, dashed, ->] (path4136-9) edge (path4136-6);
		\draw[draw=black, dashed, ->] (path4136-35) edge (path4136-7);
		\draw[draw=black, dashed, ->] (path4136-62) edge (path4136-3);
		\draw[draw=black, dashed, ->] (path4136-35) edge (path4136);
		\end{scope}

	\end{tikzpicture}
	\caption{Example instance of \DDSClong.
	The target sequence is given below the digraph.
	Adding the dashed arcs produces a digraph with the desired degree sequence.
	}
	\label{fig:example DDSC}
\end{figure}

\begin{figure}[t]
	\centering
	\begin{tikzpicture}
	\begin{scope}
		\node (rect6685-1) at (2.25, -0.25) [align=center]{};
		\node[gnode, inner sep=2.0pt, radius=3pt, minimum size=16pt, fill=white, draw=black] (path4136-5) at (0.0, 0.0) [align=left]{};
		\node[gnode, inner sep=2.0pt, radius=3pt, minimum size=16pt, fill=white, draw=black] (path4136-62) at (2.25, 0.0) [align=left]{};
		\node[gnode, inner sep=2.0pt, radius=3pt, minimum size=16pt, fill=white, draw=black] (path4136-9) at (4.5, 0.0) [align=left]{};
		\node[gnode, inner sep=2.0pt, radius=3pt, minimum size=16pt, fill=white, draw=black] (path4136-7) at (1.125, 1.25) [align=left]{};
		\node[gnode, inner sep=2.0pt, radius=3pt, minimum size=16pt, fill=white, draw=black] (path4136-35) at (3.375, 1.25) [align=left]{};
		\node[gnode, inner sep=2.0pt, radius=3pt, minimum size=16pt, fill=white, draw=black] (path4136) at (0.0, 2.5) [align=left]{};
		\node[gnode, inner sep=2.0pt, radius=3pt, minimum size=16pt, fill=white, draw=black] (path4136-3) at (2.25, 2.5) [align=left]{};
		\node[gnode, inner sep=2.0pt, radius=3pt, minimum size=16pt, fill=white, draw=black] (path4136-6) at (4.5, 2.5) [align=left]{};
		\node (rect6685) at (2.25, 3.75) [align=center]{\normalsize 1-anonymous degree sequence:\\\normalsize $0\choose 1$, $1 \choose 2$, $1 \choose 0$, $2 \choose 2$, $2 \choose 2$, $0 \choose 1$, $2 \choose 0$, $1 \choose 1$};
		\draw[draw=black, ->] (path4136-5) edge (path4136-7);
		\draw[draw=black, ->] (path4136-3) edge (path4136-35);
		\draw[draw=black, ->] (path4136-35) edge (path4136-6);
		\draw[draw=black, ->] (path4136-9) edge (path4136-35);
		\draw[draw=black, ->] (path4136-35) edge (path4136-62);
		\draw[draw=black, ->] (path4136-7) edge (path4136-62);
		\draw[draw=black, ->] (path4136-7) edge (path4136-9);
		\draw[draw=black, ->] (path4136) edge (path4136-3);
		\draw[draw=black, ->] (path4136-3) edge (path4136-7);
	\end{scope}	
	\draw[
        -triangle 90,
        line width=0.5mm,
        postaction={draw, line width=2mm, shorten >=2mm, -}
    ] (5.25,1.25) -- (5.75,1.25);
% 	\node (A) at (5.75,1.25) {};
% 	\node (B) at (6.25,1.25) {};
% 	\draw[->] (A) edge (B);
	
	\begin{scope}[shift={(6.5,0)}]
		\node (rect6685-1) at (2.25, -0.25) [align=center]{};
	\node[gnode, inner sep=2.0pt, radius=3pt, minimum size=16pt, fill=white, draw=black] (path4136-5) at (0.0, 0.0) [align=left]{$1$\\$1$};
	\node[gnode, inner sep=2.0pt, radius=3pt, minimum size=16pt, fill=white, draw=black] (path4136-62) at (2.25, 0.0) [align=left]{$2$\\$2$};
	\node[gnode, inner sep=2.0pt, radius=3pt, minimum size=16pt, fill=white, draw=black] (path4136-9) at (4.5, 0.0) [align=left]{$1$\\$1$};
	\node[gnode, inner sep=2.0pt, radius=3pt, minimum size=16pt, fill=white, draw=black] (path4136-7) at (1.125, 1.25) [align=left]{$2$\\$2$};
	\node[gnode, inner sep=2.0pt, radius=3pt, minimum size=16pt, fill=white, draw=black] (path4136-35) at (3.375, 1.25) [align=left]{$2$\\$2$};
	\node[gnode, inner sep=2.0pt, radius=3pt, minimum size=16pt, fill=white, draw=black] (path4136) at (0.0, 2.5) [align=left]{$1$\\$1$};
	\node[gnode, inner sep=2.0pt, radius=3pt, minimum size=16pt, fill=white, draw=black] (path4136-3) at (2.25, 2.5) [align=left]{$2$\\$2$};
	\node[gnode, inner sep=2.0pt, radius=3pt, minimum size=16pt, fill=white, draw=black] (path4136-6) at (4.5, 2.5) [align=left]{$1$\\$1$};
	\node (rect6685) at (2.25, 3.75) [align=center]{\normalsize 4-anonymous degree sequence:\\\normalsize $1\choose 1$, $2 \choose 2$, $1 \choose 1$, $2 \choose 2$, $2 \choose 2$, $1 \choose 1$, $2 \choose 2$, $1 \choose 1$};
	\draw[draw=black, ->] (path4136-5) edge (path4136-7);
	\draw[draw=black, ->] (path4136-3) edge (path4136-35);
	\draw[draw=black, ->] (path4136-35) edge (path4136-6);
	\draw[draw=black, ->] (path4136-9) edge (path4136-35);
	\draw[draw=black, ->] (path4136-35) edge (path4136-62);
	\draw[draw=black, ->] (path4136-7) edge (path4136-62);
	\draw[draw=black, ->] (path4136-7) edge (path4136-9);
	\draw[draw=black, ->] (path4136) edge (path4136-3);
	\draw[draw=black, ->] (path4136-3) edge (path4136-7);
	\draw[draw=black, dashed, ->] (path4136-62) edge (path4136-3);
	\draw[draw=black, dashed, ->] (path4136-62) edge (path4136-5);
	\draw[draw=black, dashed, ->, bend right=20] (path4136-6) edge (path4136);
\end{scope}
	\end{tikzpicture}
	\caption{Example instance for \DAlong.
	The input digraph is $1$-anonymous since there is only one vertex with indegree 2 and outdegree 0.
	By adding the dashed arcs we obtain a digraph which is $4$-anonymous.}
	\label{fig:example DA}
\end{figure}
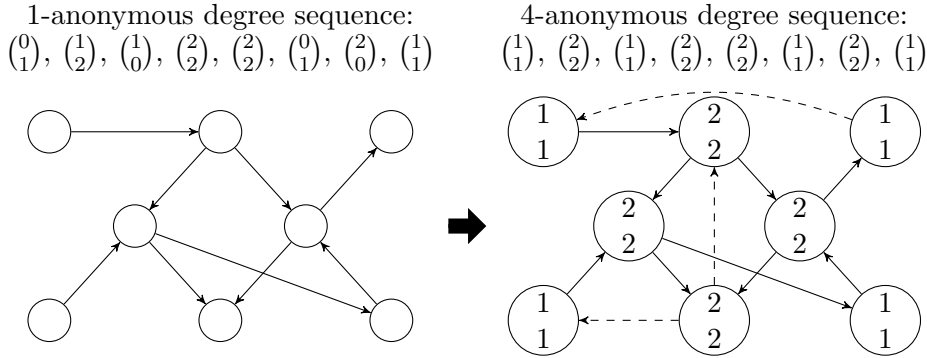

% \begin{enumerate}
%  \item  
% 	Assume we are given a directed network representing a system's current state.
% 	Then, each individual node might have certain desired states of connectivity in terms of the numbers of in- and outgoing arcs which we want to satisfy by inserting arcs between the nodes.
% 	For instance, in a peer-review network we have an arc from one author reviewing a paper of another author. 
% 	Depending on research experience, the authors might have different requests with respect to the number of own papers to be reviewed by others and other papers which they are reviewing. 
% 	This leads to the \DDCClong problem as studied in Section~\ref{sec:DDCC}.
%   
%  \item 
% 	Assume that we have two different data sources: A network which is an incomplete measurement of some unreliable source and the true degree sequence of the target network.
% 	The goal is to reconstruct the original network by inserting arcs such that we obtain the target degree sequence (in a sense, the network matches the given degree sequence).
% 	In the presence of labeled input networks this might for example reveal communication patterns between users in social networks.
% 	The corresponding problem is called \DDSClong and studied in Section~\ref{sec:DDSC}.
% 
%  \item
% 	Assume we want to ``$k$-anonymize'' a social network, that is, after inserting a minimum number of arcs each degree, that is, each combination of in- and outdegree, occurs either zero or at least $k$~times. 
% 	This leads to the \DAlong problem as studied in Section~\ref{sec:DA}.
% \end{enumerate}

All three problems are NP-hard.
Based on a general framework presented in Section~\ref{sec:GeneralSetting},
we derive several fixed-parameter tractability results for them,
mainly exploiting the parameter ``maximum vertex degree'' in the output digraph.
Moreover, the three problems above are special cases of the \DDCSClong problem which we will define next. 
% out-commented due to space:
Before doing so, however, 
we go into a little more detail concerning the roots 
of the underlying graph-theoretic  problems studied here.
Since early computer science and algorithmic graph theory
days, studies on graph realizability of degree 
sequences (that is, multisets of positive integers or integer pairs) 
have played 
a prominent role, being performed both for undirected graphs~\cite{EG60,Hak62} 
as well as digraphs~\cite{Che66,Ful60,KW73,HN15}. 
Lately, the graph modification view gained more and more 
attention: 
% These three problems as well as \DDCSClong are special graph modification problems:
given a graph, can it be changed by a minimum number of graph modifications such that the resulting graph adheres to 
specific constraints for its degree sequence?

In the most basic variant a degree sequence is a sequence of positive integers specifying (requested) vertex degrees for a fixed ordering of the vertices.
Typically, the corresponding computational problems are NP-hard.
In recent years, research in this direction focused on undirected graphs~\cite{FNN16,Gol15,GM17,HNNS15,MS12,MT09}.
In this work, we investigate parameterized algorithms on digraphs.
As \citet{GY08} observed, much less is known about the structure of digraphs than that of undirected graphs, making the design of parameterized algorithms for digraphs more challenging.
In particular, we present a general framework for a class of degree sequence modification problems, focusing on the case of arc insertions (that is, completion problems).

The most general degree completion problem for digraphs we consider in this work is as follows.

\begin{problemdef}
        \problemtitle{\DDCSClong(\DDCSC)}
        \probleminput{A digraph~$D=(V,A)$, a non-negative integer~$s$, a ``degree list function'' $\tau\colon V \rightarrow 2^{\{0,\ldots,r\}^2}$, and a ``sequence property''~$\Pi$.}
        \problemquestion{Is it possible to obtain a digraph $D'$ by inserting at most~$s$ arcs into~$D$ such that the degree sequence of~$D'$ fulfills~$\Pi$ and $\deg_{D'}(v)\in\tau(v)$ for all~$v\in V$?}
\end{problemdef}

We emphasize that there are two types of constraints---one (specified by the function~$\tau$, which gives us the in- and outdegrees) for the individual vertices and one (specified by~$\Pi$) for the whole list of degree tuples. For instance, 
a common~$\Pi$ as occurring in the context of data privacy applications 
is to request that the list is \emph{$k$-anonymous}, that is, 
every combination of in- and outdegree that occurs in the list occurs at least $k$~times (see also \Cref{fig:example DA}).

Since \DDCSC{} and its special cases as studied here all turn out 
to be NP-hard~\cite{Mil15,Kos15}, a parameterized complexity analysis
seems the most natural fit for understanding the computational
complexity landscape of these kinds of problems---this has also 
been observed in the above mentioned studies for the undirected case.
Our main findings are mostly on the positive side. That is,
although seemingly more intricate to deal with due to the 
existence of in- and outdegrees, many positive algorithmic results 
which hold for undirected graphs can also be achieved for 
digraphs (albeit using different techniques).
In particular, we present a maximum-flow-based framework that,
together with the identification and solution of certain number problems,
helps to derive several fixed-parameter tractability results with 
respect to the parameter maximum possible in- or outdegree~\degpara in any solution digraph.
Notably, the corresponding result in the undirected case was 
based on $f$-factor computations~\cite{FNN16} which do not transfer 
to the directed case, and, vice versa, the flow computation 
approach we present for the directed case seemingly does not transfer to the 
undirected case. 
For special cases of \DDCSC{}, we can move further and even derive 
some polynomial-size problem kernels, again for the parameter~\degpara.

We consider the parameter~\degpara for the following reasons.
First, it is always at most~$r$, a natural parameter in the input.
Second, in combination with~$\Pi$, we might get an even smaller upper bound for~\degpara.
Third, bounded-degree graphs are well studied and our work extends this since we only require~\degpara to be small, not to be constant.
% out-commented due to space:
Fourth, in practice, the maximum degree is often significantly smaller than the number of vertices:
\citet{LH08} studied a huge instant-messaging network (180 million vertices) with maximum degree~600.
% For the DBLP co-author graph\footnote{In this graph the vertices represent the authors and an edge indicates that the two corresponding authors are co-authors of at least one paper. The dataset and a corresponding documentation are available online (\url{http://dblp.uni-trier.de/xml/}).}
% generated in February~2012 containing more than 715,000~vertices one has a maximum degree of~804 and an H-index of~208, that is, there are not more than~208 vertices with degree larger than~208~\cite{HNNS15}.
Furthermore, in the context of anonymization it can empirically be observed that the maximum degree will not increase during the anonymization process~\cite{HHN14}. 
Thus, the parameter~\degpara is interesting when studying kernelization as we do.

\subsection{Related Work}

Most of the work on graph modification
problems for realizing degree constraints has focused on undirected 
graphs~\cite{FNN16,Gol15,GM17,HNNS15,MS12,MT09}. Closest to our work
is the framework for deriving polynomial-size problem kernels
for undirected degree sequence completion problems~\cite{FNN16},
which we complement by our results for digraphs. 
Generally speaking, we can derive similar results, but the technical
details differ and the landscape of problems is richer 
in the directed case.
As to digraph modification problems in general, we are aware of surprisingly 
little work. 
We mention work studying arc insertion for 
making a digraph transitive~\cite{WKNU12} or 
for making a graph 
Eulerian~\cite{DMNW13}, both employing the toolbox of parameterized
complexity analysis. 
Somewhat related is also work about the insertion of edges into a mixed graph to satisfy local edge-connectivity constraints~\cite{FJ95}
or about orienting edges in a partially oriented graph to make it an 
oriented graph~\cite{BHZ18}.

\subsection{Our Results}
In \cref{sec:GeneralSetting}, we present our general framework for \DDCSC{}, which is a two-stage approach based on flow computations.
To this end, we identify a specific pure number problem arising from the degree constraints.
We show that, if this number problem is fixed-parameter tractable with respect to the largest possible integer in the output, then \DDCSC{} is fixed-parameter tractable with respect to~\degpara.
Next, presenting applications of the framework, in \cref{sec:DDCC}, we show that if there is no constraint~$\Pi$ concerning the degree sequence (that is, we deal with the \DDCClong problem), then we not only obtain fixed-parameter 
tractability but also a polynomial-size problem kernel for parameter~\degpara.
Then, in \cref{sec:DDSC} we show an analogous result if there is exactly one specified degree sequence to be realized (\DDSClong).
Finally, in \cref{sec:DA}, we show that if we request the degree
sequence to be $k$-anonymous (that is, \DAlong), then we can at least derive a polynomial-size
problem kernel for the combined parameter~$(s,\Delta_D)$, where~$\Delta_D$ denotes the maximum in- or outdegree of the input digraph~$D$.
Also, we take a first step outlining the limitations of our framework for digraphs.
In contrast to the undirected case (which is polynomial-time solvable~\cite{LT08}), the corresponding number problem of \DAlong surprisingly is weakly NP-hard and, thus, presumably not polynomial-time solvable.
A summary of our results is provided in \cref{tab:summary}.

\begin{table}
  \centering
  \caption{Summary of our results for the three problems we studied (indicated by (digraph)) and the corresponding number problems (indicated by (number)). The parameters are defined as follows:
  $n$ is the number of vertices, $s$ is the maximum number of added arcs, $r$ is the maximum target in- or outdegree of a vertex, \degpara is the maximum in- or outdegree in any solution digraph,
  $\Delta_D$ is the maximum in- or outdegree in the input digraph, $\xi$ is the largest possible integer in the output sequence, and $k$ is the level of anonymity.}
  \label{tab:summary}
  \begin{tabularx}{\textwidth}{lp{1.5cm}Xr}
	\toprule
	Problem & & Result & Reference\\
	% \midrule
	% \DDCSC (number)  & FPT w.r.t.\ largest input number\\
	% \DDCSC (digraph) & FPT w.r.t.\ \degpara\\
	\midrule
	\multirow{3}{*}{\DDCC} & (number)  & $O(n(sr)^2)$-time solvable & \Cref{lem:nDDCC-poly}\\
                & (digraph) & $O(s(\degpara)^3)$-vertex kernel in $O(m+|\tau|+r^2)$ time & \Cref{thm:DDCCkrKernel}\\
                &  & $O((\degpara)^5)$-vertex kernel in $O(m+ns^3r^2)$ time & \Cref{cor:DDCCkernel}\\
	\midrule
	\multirow{3}{*}{\DDSC} & (number) & $O(n^{2.5})$-time solvable & \Cref{lem:TSC-poly}\\
                & (digraph) & $O(s(\degpara)^3)$-vertex kernel in $O(n+m+(\degpara)^2)$ time & \Cref{thm:DDSC_kDelta-kernel}\\
                & & $O((\degpara)^5)$-vertex kernel in $O(sn^{2.5})$ time & \Cref{cor:DDSCkernel}\\
	\midrule
	\multirow{5}{*}{\DA} & (number)  & weakly NP-hard & \Cref{thm:nDAhard}\\
				& &FPT w.r.t.\ $\xi$ & \Cref{thm:nDA-FPT}\\
                & (digraph) & FPT w.r.t.\ \degpara & \Cref{cor:DDAfpt}\\
                & & $O(s\Delta_D^5)$-vertex kernel in $O(\Delta_D^{10}s^2+\Delta_D^3sn)$ time & \Cref{thm:DirDegAnon-kernel-sDelta}\\
                & & FPT w.r.t.\ $(k,\Delta_D)$ & \Cref{cor:DA-FPT-deltastern}\\
	\bottomrule
\end{tabularx}
\end{table}

\section{Preliminaries} \label{notation}

\paragraph{Notation.}
We consider \emph{digraphs} (without multiarcs or loops) $D=(V,A)$ with $n \coloneqq |V|$ and $m \coloneqq |A|$.
For a vertex~$v\in V$, $\deg_D^-(v)$ denotes the \emph{indegree} of~$v$, that is, the number of incoming arcs of~$v$. Correspondingly, $\deg_D^+(v)$ denotes the \emph{outdegree}, that is, the number of outgoing arcs of~$v$.
We define the \emph{degree}~$\deg_D(v):=(\deg_D^-(v),\deg_D^+(v))$.
The set $V(A') \coloneqq \{ v \in V  \mid ((v,w) \in A' \vee (w,v) \in A') \wedge w \in V \}$ contains all vertices incident to an arc in $A'\subseteq V^2$.
For a set of arcs $A'\subseteq V^2$, $D+A'$ denotes the digraph~$(V,A\cup A')$, while $D[A']$ denotes the  subdigraph~$(V(A'), A')$.
Analogously, for a set of vertices $V' \subseteq V$, $D[V']$ denotes the induced subdigraph~$(V',A\cap (V')^2)$ which only contains the vertices~$V'$ and the arcs between vertices from~$V'$.
The set $N_{D}^+(v) \coloneqq \{ w \in V \mid (v,w) \in A\}$ denotes the set of \emph{outneighbors} of~$v$. Analogously, $N_{D}^-(v) \coloneqq \{ w \in V \mid (w,v) \in A\}$ denotes the set of \emph{inneighbors}.
Furthermore, we define the maximum indegree $\Delta^-_D\coloneqq \max_{v \in V} \deg_D^-(v)$, the maximum outdegree $\Delta^+_D\coloneqq \max_{v \in V} \deg_D^+(v)$, and~$\Delta_D:=\max\{\Delta_D^+,\Delta_D^-\}$.

A \emph{digraph degree sequence} $\sigma=\{(d_1^-,d_1^+),\ldots,(d_n^-,d_n^+)\}$ is a multiset of nonnegative integer tuples, where $d_i^-,d_i^+\in\{0,\ldots,n-1\}$ for all~$i\in\{1,\ldots,n\}$. We define
\begin{align*}
  \Delta^-_\sigma&:=\max\{d_1^-,\ldots,d_n^-\},\\
  \Delta^+_\sigma&:=\max\{d_1^+,\ldots,d_n^+\}, \text{ and}\\
  \Delta_\sigma&:=\max\{\Delta^-_\sigma,\Delta^+_\sigma\}.
\end{align*}
For a digraph~$D=(\{v_1,\ldots,v_n\},A)$ we denote by $\sigma(D):=\{\deg_D(v_1)$, $\ldots$, $\deg_D(v_n)\}$, the digraph degree sequence of~$D$.
Let~$d = (d^-,d^+)$ be a nonnegative integer tuple.
For a digraph~$D$, the \emph{block}~$B_D(d)$ \emph{of degree~$d$} is the set of all vertices having degree~$d$, formally~$B_D(d) := \{v \in V \mid \deg_D(v) = d\}$.
We define~$\lambda_D(d)$ as the number of vertices in~$D$ with degree~$d$, that is, $\lambda_D(d) := |B_D(d)|$.
Similarly, we define~$B_\sigma(t)$ as the multiset of all tuples equal to~$t$
and $\lambda_\sigma(t)$ as the number of occurrences of the tuple~$t$ in the multiset~$\sigma$.
For two integer tuples~$(x_1,y_1)$, $(x_2,y_2)$, we define the sum~$(x_1,y_1)+(x_2,y_2):=(x_1+x_2,y_1+y_2)$.

\paragraph{Parameterized Algorithmics. \cite{Cyg15,DF13,FG06,Nie06}} We assume the reader to be familiar with classical complexity theory concepts such as polynomial-time reductions and (weak) NP-hardness~\cite{GJ79, AB09}.
An instance~$(I,k)$ of a parameterized problem~$L\subseteq \Sigma^*\times\mathbb{N}$ consists
of the classical input~$I$ and a \emph{parameter}~$k$. A parameterized problem~$L$ is called \emph{fixed-parameter tractable} (fpt) with respect to the parameter~$k$ if it can be solved in~$f(k)\cdot |I|^{O(1)}$ time, where~$f$ is a function only depending on~$k$ and~$|I|$ denotes the size of the input~$I$.
Accordingly, for a \emph{combined parameter}~$(k_1,k_2,\ldots)$, a parameterized problem is fpt if it can be solved in~$f(k_1,k_2,\ldots)\cdot |I|^{O(1)}$ time.

A \emph{kernelization} is a polynomial-time algorithm transforming a given instance~$(I,k)$ into an equivalent instance~$(I',k')$ with~$|I'|\le g(k)$ and~$k'\le h(k)$ for some functions~$g$ and~$h$, that is, $(I,k)$ is a yes-instance if and only if~$(I',k')$ is a yes-instance.
The instance~$(I',k')$ is called the \emph{problem kernel} and~$g$ denotes its size.
If~$g$ is a polynomial, then we have a \emph{polynomial(-size)} problem kernel.
It can be shown that a parameterized problem is fpt if and only if it has a problem kernel.

\section{The Framework} \label{sec:GeneralSetting}

      Our goal is to develop a framework for deriving fixed-parameter tractability for a general
      class of completion problems in directed graphs.
      To this end, recall our general setting for \DDCSC{} which is as follows.
      We are given a digraph and want to insert at most~$s$ arcs such that the vertices satisfy certain degree constraints~$\tau$,
      and, additionally, the degree sequence of the digraph fulfills a certain property~$\Pi$.
Formally, the sequence property~$\Pi$ is given as a function that maps a digraph degree sequence to~$1$ if the sequence fulfills the property and otherwise to~$0$.
We restrict ourselves to properties where the corresponding function can be encoded with only polynomially many bits in the number of vertices of the input digraph
and can be decided efficiently.\footnote{All specific properties in this work can be easily decided in polynomial time.
Indeed, in many cases even fixed-parameter tractability with respect to the maximum integer in the sequence would suffice.}
We remark that it is not always the case that there are both vertex degree constraints (as defined by~$\tau$) and degree sequence constraints~(as defined by~$\Pi$) requested.
This can be handled by either setting~$\tau$ to the trivial degree list function with~$\tau(v)=\{0,\ldots,n-1\}^2$ for all $v\in V$ or setting~$\Pi$ to allow all possible degree sequences.

In this section, we show how to derive (under certain conditions) fixed-parameter tractability with respect to the maximum possible in- or outdegree~\degpara of the \emph{output} digraph for \DDCSC.
Note that~\degpara in general is not known in advance. 
In practice, we might therefore instead consider upper bounds for~\degpara which depend on the given input. 
For example, we know that $\degpara \le \min\{r,\Delta_D+s\}$ for any yes-instance since we insert at most~$s$ arcs into~$D$.
Clearly, \degpara might also be upper-bounded depending on~$\Pi$ (or even depending on~$r$, $s$, $\Delta_D$, and~$\Pi$) in some cases.
Our generic framework consists of two main steps: First, we prove fixed-parameter tractability with respect to the combined parameter~$(s,\Delta_D)$ in \cref{sec:GeneralFPTsr}.
This step generalizes ideas for the undirected case~\cite{FNN16}.
Note that~$\Delta_D\le \degpara$ trivially holds for yes-instances.
Second, we show in \cref{sec:GeneralFPTr} how to upper-bound the number~$s$ of arc insertions polynomially in~$\degpara$ by solving a certain problem specific numerical problem.
For this step, we develop a new key argument based on a maximum flow computation (the undirected case was based on $f$-factor arguments).

\subsection{Fixed-parameter tractability with respect to~$(s,\Delta_D)$}
\label{sec:GeneralFPTsr}
We show that \DDCSC is fixed-parameter tractable with respect to the combination of
the maximum number~$s$ of arcs to insert and the maximum in- or outdegree~$\Delta_D$ of the input digraph~$D$.
The basic idea underlying this result is that two vertices~$v$ and~$w$ with~$\deg_D(v)=\deg_D(w)$ and~$\tau(v) = \tau(w)$ are interchangeable.
Accordingly, we will show that it suffices to consider only a bounded number of vertices with the same ``degree properties''.
In particular, if there is a solution, then there is also a solution that only inserts arcs between a properly chosen subset of vertices of bounded size.
To formalize this idea, we introduce the notion of an \emph{$\alpha$-block-type set} for some positive integer~$\alpha$.

To start with, we define the types of a vertex via the numbers of arcs that~$\tau$ allows to add to this vertex.
Let~$(D,s,\tau,\Pi)$ be a \DDCSC instance.
A vertex $v$ is of \emph{type} $t \in \{0,\ldots,\degpara\}^2$ if $\deg_D(v)+t \in \tau(v)$.
Observe that one vertex can be of several types.
The subset of~$V(D)$ containing all vertices of type~$t$ is denoted by~$T_{D,\tau}(t)$.
A vertex~$v$ of type~$(0,0)$ (that is, $\deg_D(v)\in\tau(v)$) is called \emph{satisfied}.
A vertex which is not satisfied is called \emph{unsatisfied}.
We next define our notion of $\alpha$-block-type sets and its variants.

\begin{definition} \label[definition]{def:alpha-block-type-set}
Let~$\alpha$ be a positive integer and let~$U\subseteq V(D)$ denote the set of all unsatisfied vertices in~$D$. A vertex subset~$C \subseteq V(D)$ with~$U\subseteq C$ is called
\begin{itemize}
	\item \emph{$\alpha$-type set} if, for each type~$t \neq (0,0)$, $C$ contains exactly $\min \{|T_{D,\tau}(t)\setminus U|, \alpha\}$ satisfied vertices of type~$t$;
	\item \emph{$\alpha$-block set} if, for each degree~$d \in \sigma(D)$, $C$ contains exactly $\min \{|B_D(d)\setminus U|, \alpha\}$ satisfied vertices with degree~$d$;
	\item \emph{$\alpha$-block-type set} if, for each degree~$d \in \sigma(D)$ and each type~$t \neq (0,0)$, $C$~contains exactly $\min \{|(B_D(d) \cap T_{D,\tau}(t))\setminus U|, \alpha\}$ satisfied vertices of degree~$d$ and type~$t$.
\end{itemize}
\end{definition}

As a first step, we prove that these sets defined above can be computed efficiently.% (see \cref{A:lem:alphaSetLinTime} for details).

\begin{lemma} \label[lemma]{lem:alphaSetLinTime}
	An $\alpha$-type/$\alpha$-block/$\alpha$-block-type set~$C$ as described in \cref{def:alpha-block-type-set} can be computed in~$O(m + |\tau| + r^2)$ / $O(m + n + \Delta_D^2)$ / $O(m + |\tau| + \Delta_D^2 r^2)$ time. 
\end{lemma}

\begin{proof}
  To compute an $\alpha$-block-type set, we start with an empty set~$C:=\emptyset$ and
  for each possible vertex degree~$d$ and each possible vertex type~$t$, we initialize a counter~$x(d,t):=0$.
  We then iterate over all vertices~$v\in V(D)$.
  If~$v$ is unsatisfied, that is~$\deg_D(v)\not\in\tau(v)$, then we add~$v$ to~$C$.
  If~$v$ is satisfied, then for each type~$t\neq (0,0)$ with~$\deg_D(v) + t \in \tau(v)$, we increase the counter~$x(\deg_D(v),t)$ by one and add~$v$ to~$C$ if $x(\deg_D(v),t)<\alpha$.
	This can be done in~$O(m+|\tau| + \Delta_D^2 r^2)$ time.
	The other two cases of computing an $\alpha$-block set or an $\alpha$-type set can be done in a similar fashion.\qed
\end{proof}

We move on to the crucial lemma stating that a solution (that is, a set of arcs), if existing, can always be found between vertices of an $\alpha$-block-type set~$C$ given that~$C$ contains ``enough'' vertices of each degree and type.
Here, enough means $\alpha:=2s(\Delta_D+1)$.
\begin{lemma} \label[lemma]{lem:AlphaSet}
	Let~$(D,s,\tau,\Pi)$ be a \DDCSC instance and let~$C\subseteq V(D)$ be a $2s(\Delta_D+1)$-block-type set.
	If~$(D,s,\tau,\Pi)$ is a yes-instance, then there exists a solution~$A^*\subseteq C^2$ for~$(D,s,\tau,\Pi)$, that is, $|A^*|\le s$, $\sigma(D+A^*)$ fulfills~$\Pi$, and~$\deg_{D+A^*}(v) \in \tau(v)$ for all~$v \in V(D)$.
\end{lemma}

\begin{proof}
	Let~$A' \subseteq V(D)^2 \setminus A(D)$ be a solution for $(D,s,\tau,\Pi)$ that minimizes the number of vertices not in~$C$, that is, $|V(A') \setminus C|$ is minimum. 
	The solution~$A'$ exists since~$(D,s,\tau,\Pi)$ is a yes-instance.
	If~$V(A')\subseteq C$, then we are done.
	Hence, we assume that there exists a vertex~$v$ in~$V(D) \setminus C$ which is incident to at least one arc in~$A'$.
	Let~$V^-_v := \{ u \mid (u,v)\in A'\}$ and let $V^+_v := \{w \mid  (v,w)\in A'\}$ be the set of in- respectively outneighbors of~$v$ in~$A'$.
	Furthermore, let~$d := \deg_D(v)$ and~$t := (|V^-_v|,|V^+_v|)$.
	Thus, $v$ has degree~$d$ and is of type~$t$.
	By definition of~$C$, it follows that~$|B_D(d) \cap T_{D,\tau}(t)| > 2s(\Delta_D+1)$.

	Now, we claim that there is a vertex~$v^* \in (B_D(d) \cap T_{D,\tau}(t) \cap C) \setminus V(A')$ such that we can replace~$v$ with~$v^*$ in the solution.
	More precisely, in all arcs of~$A'$ we want to replace $v$ by~$v^*$, that is, we obtain a new arc set
	$A^*:= \{(u,w) \in A' \mid u \neq v \wedge w \neq v \} \cup \{(u,v^*)\mid u \in V^-_v\} \cup \{(v^*,w)\mid w\in V^+_v\}.$
	Since we cannot insert arcs that already exist in the input digraph~$D$, we need that~$N_D^-(v^*) \cap V^-_v = \emptyset$ and~$N_D^+(v^*) \cap V^+_v = \emptyset$.
	Observe that such a vertex~$v^*$ exists: 
	Since each of the at most~$s$ vertices in~$V^+_v \cup V^-_v$ has at most~$\Delta_D$ incoming and~$\Delta_D$ outgoing arcs, it follows that at most~$s \cdot 2 \Delta_D$ vertices in~$B_D(d) \cap T_{D,\tau}(t) \cap C$ can have an arc from or to a vertex in~$V^+_v \cup V^-_v$.
	Furthermore, since~$|A'| \le s$, it follows that at most~$2s - 1$ vertices in~$B_D(d) \cap T_{D,\tau}(t) \cap C$ are incident to an arc in~$A'$ (the minus one comes from the fact that~$v$ is incident to at least one arc in~$A'$). 
	By definition of~$C$, it follows that~$|B_D(d) \cap T_{D,\tau}(t) \cap C| \ge 2s(\Delta_D + 1) > s \cdot 2 \Delta_D + 2s -1$.
	Hence, there is at least one vertex~$v^* \in B_D(d) \cap T_{D,\tau}(t) \cap C$ that is not adjacent to any vertex in~$V^+_v \cup V^-_v$ and not incident to any arc in~$A'$.
	Thus, we can replace~$v$ by~$v^*$.
	
	We now show that~$A^*$ is still a solution:
	First, observe that $\sigma(D+A') = \sigma(D+A^*)$ and, hence, $\sigma(D+A^*)$ fulfills $\Pi$.
	Second, observe that~$\deg_{D+A^*}(v) \in \tau(v)$ since~$v \notin C$, which implies that~$v$ was satisfied.
	Further, $\deg_{D+A^*}(v^*) \in \tau(v^*)$ since~$v^*$ is of type~$t$.
	Hence, $A^*$ is a solution and~$|V(A') \setminus C| > |V(A^*) \setminus C|$, a contradiction to the assumption that~$A'$ was a solution minimizing this value.\qed
\end{proof}

If there are no restrictions on the resulting degree sequence (as it is the case for the \DDCClong problem (\DDCC) in \cref{sec:DDCC}), then we can replace the $2s(\Delta_D+1)$-block-type set in \cref{lem:AlphaSet} by a $2s(\Delta_D+1)$-type set:
\begin{lemma}\label[lemma]{lem:DDCC-type-set}
	Let~$(D,s,\tau)$ be a \DDCC instance and let~$C\subseteq V(D)$ be a $2s(\Delta_D+1)$-type set.
	If~$(D,s,\tau)$ is a yes-instance, then there exists a solution~$A^*\subseteq C^2$ for~$(D,s,\tau)$, that is, $|A^*|\le s$ and~$\deg_{D+A^*}(v) \in \tau(v)$ for all~$v \in V(D)$.
\end{lemma}

Similarly, if there are no restrictions on the individual vertex degrees, that is, $\tau$ is the degree list function~$\tau(v)=\{0,\ldots,n-1\}^2$ for all~$v\in V(D)$,
then we can replace the~$2s(\Delta_D+1)$-block-type set by a $2s(\Delta_D+1)$-block set.

\begin{lemma}\label[lemma]{lem:PiDDSC-block-set}
	Let~$(D,s,\tau,\Pi)$ be a \DDCSC instance with~$\tau(v) = \{0,\ldots,n-1\}^2$ for all~$v\in V(D)$ and let~$C\subseteq V(D)$ be a $2s(\Delta_D+1)$-block set.
	If~$(D,s,\tau,\Pi)$ is a yes-instance, then there exists a solution~$A^*\subseteq C^2$ for~$(D,s,\tau,\Pi)$, that is, $|A^*|\le s$ and $\sigma(D+A^*)$ fulfills~$\Pi$.
\end{lemma}

\cref{lem:AlphaSet} implies a fixed-parameter algorithm by providing a bounded search space for possible solutions, namely any $2s(\Delta_D + 1)$-block-type set~$C$.

\begin{theorem}\label{thm:PiEAfptDeltak}
  If deciding $\Pi$ is fixed-parameter tractable with respect to the maximum integer in the input sequence, then \DDCSC is fixed-parameter tractable with respect to~$(s,\Delta_D)$.
\end{theorem}

\begin{proof}
  Given a \DDCSC instance~$(D,s,\tau,\Pi)$, we first check in polynomial time whether there are more than~$2s$ unsatisfied vertices in~$D$. If this is the case, then we have a no-instance, since we can change the degrees of at most~$2s$ vertices by inserting at most~$s$ arcs.
  Otherwise, we compute a $2s(\Delta_D + 1)$-block-type set~$C$ in polynomial time (\cref{lem:alphaSetLinTime}).
  By~\cref{lem:AlphaSet}, we know that it is sufficient to search for a solution within the vertices of~$C$. Hence, we simply try out all possible arc sets~$A' \subseteq C^2$ of size at most~$s$ and check whether in one of the cases the vertex degrees and the degree sequence of~$D+A'$ satisfy the requirements~$\tau$ and~$\Pi$.
  Since~$C$ contains at most~$2s$ unsatisfied vertices and at most~$2s(\Delta_D + 1)\cdot (\Delta_D+1)^2(\degpara)^2$ satisfied vertices, and since~$\degpara \le \Delta_D + s$, there are at most~$O(2^{(2s+2s(\Delta_D+1)^3(\Delta_D+s)^2)^2})$ possible subsets of arcs to insert.
  Checking whether~$\tau$ is satisfied can be done in polynomial time and deciding whether~$\Pi$ holds for $\sigma(D+A')$ is by assumption fixed-parameter tractable with respect to the largest integer, which is at most~$\Delta_D+s$.
Thus, overall, we obtain fixed-parameter tractability with respect to~$(s,\Delta_D)$.\qed
\end{proof}

\subsection{Bounding the solution size~$s$ polynomially in~\degpara}
\label{sec:GeneralFPTr}
This subsection constitutes the major part of our framework.
The rough overall scheme is analogous to the undirected case as described by \citet{FNN16}.
By dropping the graph structure and solving a simpler problem-specific number problem on the degree sequence of the input digraph, we show how to solve \DDCSC instances with ``large'' solutions provided that we can solve the associated number problem efficiently.
The number problem is defined so as to simulate the insertion of arcs to a digraph on an
integer tuple sequence.
Note that inserting an arc increases the indegree of a vertex by one and increases the outdegree of another vertex by one.
Inserting~$s$ arcs can thus be represented by increasing the tuple entries in the degree sequence by an overall value of~$s$ in each component.
Formally, the corresponding number problem (abbreviated as \nDDCSC) is defined as follows.

\begin{problemdef}
    \problemtitle{\nDDCSClong}
    \probleminput{A sequence $\sigma = (c_1,d_1),\ldots,(c_n,d_n)$ of~$n$ nonnegative integer tuples,
                  a positive integer~$s$,
                  a ``tuple list function'' $\tau\colon \{1,\ldots,n\} \rightarrow 2^{\{0,\ldots,r\}^2}$,
                  and a sequence property~$\Pi$.}
    \problemquestion{Is there a sequence $\sigma'=(c'_1,d'_1),\ldots,(c'_n,d'_n)$
                     such that $\sum_{i=1}^n (c'_i-c_i) = \sum_{i=1}^n (d'_i-d_i) = s$,
                     $c_i \le c'_i$, $d_i \le d'_i$, and $(c'_i, d'_i)\in\tau(i)$ for all $1 \le i \le n$, and
                     $\sigma'$ fulfills~$\Pi$?}
\end{problemdef}

%Note that
If we plug the degree sequence of a digraph into \nDDCSC, then an integer tuple~$(c_i',d_i' )$
of a solution tells us to add~$x_i:=c_i'-c_i$ incoming arcs and~$y_i:=d_i'-d_i$ outgoing arcs to the vertex~$v_i$.
We call the tuples~$(x_i,y_i)$ \emph{demands}.
Having computed the demands, we can then try to solve our original \DDCSC instance by searching for a set of arcs to insert that exactly fulfills the demands.
Such an arc set, however, might not always exist.
Hence, the remaining problem is to decide whether it is possible to realize the demands in the given digraph. The following lemma shows (using flow computations) that this is in fact always possible if the number~$s$ of arcs to insert is large compared to~\degpara.

\begin{lemma}\label[lemma]{lem:factor}
  Let~$D=(V=\{v_1,\ldots,v_n\},A)$ be a digraph and let~$x_1,\ldots,x_n$, $y_1,\ldots,y_n$, and~\degpara be nonnegative integers such that
  \begin{compactenum}[(I)]
    \item\label[cond]{cond:maxdegbound} $\degpara\le n-1$, 
    \item\label[cond]{cond:indegbound} $\deg_D^-(v_i)+x_i \le \degpara$ for all~$i\in\{1,\ldots, n\}$,
    \item\label[cond]{cond:outdegbound} $\deg_D^+(v_i)+y_i \le \degpara$ for all~$i\in\{1,\ldots, n\}$,   
    \item\label[cond]{cond:balance} $\sum_{i=1}^nx_i=\sum_{i=1}^ny_i=:s$, and
    \item\label[cond]{cond:minsize} $s > 2(\degpara)^2$.
  \end{compactenum}

  Then, there exists an arc set~$A'\subseteq V^2 \setminus A$ of size~$s$ such that for the digraph~$D':=D+A'$ it holds $\deg_{D'}(v_i)=\deg_D(v_i)+(x_i,y_i)$ for all $v_i\in V$.
  Moreover, the set~$A'$ can be computed in~$O(n^3)$ time. %in~$O(n(n^2-m))$ time.
\end{lemma}

\begin{proof}
The proof is based on a flow network which we construct such that the corresponding maximum flow (for details of network flow theory, refer to the book by~\citet{AMO93}) yields the set~$A'$ of arcs to be inserted in~$D$ in order to obtain our target digraph~$D'$. 
\begin{construction}\label[construction]{flowConst}
  We build a flow network~$N=(V_N,A_N)$ according to the following steps.
  \begin{itemize}
    \item Add a \emph{source} vertex~$v_s$ and a \emph{sink} vertex~$v_t$ to~$N$;
    \item for each vertex $v_i\in V$, add two vertices~$v_i^+$, $v_i^-$ to~$N$;
    \item for each $i \in \{1,\dots,n\}$, insert the arc $(v_s,v_i^+)$ with capacity $y_i$;
    \item for each $i \in \{1,\dots,n\}$, insert the arc $(v_i^-,v_t)$ with capacity $x_i$;
    \item for each $(v_i,v_j)\in V^2\setminus A$ with~$i\neq j$, insert the arc $(v_i^+,v_j^-)$ with capacity one.
  \end{itemize}
\end{construction}

% \begin{SCfigure}[2][t!]
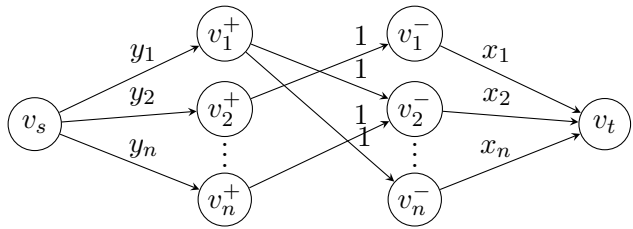
\begin{figure}[t!]
  \centering
  \begin{tikzpicture}[>=stealth,scale=1.0, transform shape]
    % sink
    \node[draw, shape=circle, inner sep=3pt] (t-1) at (7.5,2.8) {$v_t$};

    % quelle
    \node[draw, shape=circle, inner sep=3pt] (s-1) at (0,2.8) {$v_s$};
    
    \node[draw, shape=circle, inner sep=1pt] (input-1) at (2.5,4) {$v_1^+$};
    \node[draw, shape=circle, inner sep=1pt] (input-2) at (2.5,3) {$v_2^+$};
    \node at (2.5,2.5) {$\vdots$};
    \node[draw, shape=circle, inner sep=1pt] (input-3) at (2.5,1.8) {$v_n^+$};

    \node[draw, shape=circle, inner sep=1pt] (hidden-1) at (5,4) {$v_1^-$};
    \node[draw, shape=circle, inner sep=1pt] (hidden-2) at (5,3) {$v_2^-$};
    \node at (5.0,2.5) {$\vdots$};
    \node[draw, shape=circle, inner sep=1pt] (hidden-3) at (5,1.8) {$v_n^-$};

    % draw some example arcs between both sides
    \draw [->] (input-1) -- (hidden-3) node [pos=0.8, above] {$1$};
    \draw [->] (input-1) -- (hidden-2) node [pos=0.8, above] {$1$};
    \draw [->] (input-2) -- (hidden-1) node [pos=0.8, above] {$1$};
    \draw [->] (input-3) -- (hidden-2) node [pos=0.8, above] {$1$};
    
    %arcs from source to out-vertices
    \foreach \l [count=\i] in {1,2,n}
        \draw [->] (s-1) -- (input-\i) node [pos=0.6, above] {$y_\l$}; 
    
    %arcs from in-vertices to sink
    \foreach \l [count=\i] in {1,2,n}
        \draw [->] (hidden-\i) -- (t-1) node [pos=0.4, above] {$x_\l$};
\end{tikzpicture} 
\caption{A flow network as described in \Cref{flowConst}. For each vertex $v_i$ in the digraph~$D$ there are two vertices $v_i^+$ and~$v_i^-$. We connect a vertex~$v_i^+$ to a vertex $v_j^-$ if the arc $(v_i,v_j)$ is not in $D$. Inserting the arc $(v_i,v_j)$ is then represented by setting the flow on the arc $(v_i^+,v_j^-)$ to one.
}
\label{fig:flownw}
\end{figure}
% \end{SCfigure}
The network~$N$ contains~$|V_N|\in O(n)$ vertices and~$|A_N|\in O(n^2-m)$ arcs (since $m\le n^2 -n$, we also have~$|A_N|\in\Omega(n)$) and can be constructed in~$O(n^2)$ time.
See \Cref{fig:flownw} for an illustration.
Inserting an arc $(v_i,v_j)$ in~$D$ corresponds to sending flow from $v_i^+$ to $v_j^-$. Since, by definition,
each vertex~$v_i^+$ will only receive at most~$y_i$ flow from~$v_s$ and each vertex $v_i^-$ will send at most $x_i$ flow to $v_t$, we cannot insert more than~$s$ arcs (\cref{cond:balance}).

We claim that for $s>2(\degpara)^2$ (\cref{cond:minsize}), the maximum flow in the network is indeed $s$.
To see this, let $V_N^+\coloneqq \{ v_i^+ \in V_N \mid i \in \{1,\dots,n\}\}$ and let $V_N^-\coloneqq \{ v_i^- \in V_N \mid i \in \{1,\dots,n\}\}$.
	In the following, a vertex $v_i^+ \in V_N^+$ ($v_j^- \in V_N^-$) is called \emph{saturated} with respect to a flow~$f\colon A_N \rightarrow \mathbb{R}^+$, if $f(v_s,v_i^+) = y_i$ ($f(v_j^-,v_t) = x_j$).
	Suppose that the maximum flow~$f$ has a value less than $s$.
	Then, there exist non-saturated vertices $v_i^+ \in V_N^+$ and $v_j^- \in V_N^-$.
	Let $X \subseteq V_N^-$  be the vertices to which $v_i^+$ has an outgoing arc in the residual graph and let $Y \subseteq V_N^+$ be the vertices which have an outgoing arc to $v_j^-$ in the residual graph.
	Observe that $\deg_N^+(v_i^+) = n-1-\deg_D^+(v_i)$ and $\deg_N^-(v_j^-) = n-1-\deg_D^-(v_j)$.
	Consequently, $|X| > n-1-\deg_D^+(v_i) - y_i \ge n-1-\degpara$ holds due to \cref{cond:outdegbound}.
	Since $v_i^+$ is not saturated, we know that $|X| \ge n-\degpara\ge 1$ (due to \cref{cond:maxdegbound}).
	By the same reasoning (using \cref{cond:indegbound,cond:maxdegbound}) it follows that $|Y| \ge n-\degpara\ge 1$.

	Remember that $f$ is a flow of maximum value. Hence, each vertex in~$X$ and each vertex in~$Y$ is saturated. 
	Otherwise, there would be an augmenting path in the residual graph, contradicting our assumption of $f$ being maximal. 
	If a vertex~$x \in X$ would receive flow from a vertex~$y \in Y$, then this would imply a backward arc in the residual graph resulting in an augmenting path $v_s \rightarrow v_i^+ \rightarrow x \rightarrow y \rightarrow v_j^- \rightarrow v_t$, again contradicting our maximality assumption for~$f$. 
	Thus, we can conclude that all the flow that goes into $X$ has to come from the remaining vertices in $V_N^+\setminus (Y\cup\{v_i^+\})$. This set has size at most $n - |Y| \le n-(n-\degpara)=\degpara$.
	But since $y_\ell \leq \degpara$ for all $\ell \in \{1,\dots,n\}$ (by \cref{cond:outdegbound}), those~$\degpara$ vertices can cover at most a flow of value~$(\degpara)^2$ and, hence, 
	\begin{align}
		\sum_{v_i^- \in X} x_i & \le \sum_{v_i^+ \in V_N^+ \setminus (Y \cup \{v_i^+\})} y_i \le (\degpara)^2. \label{eq:X-has-little-flow}
	\end{align}
	Since $X$ is saturated, and since also $x_\ell\le \degpara$ holds for all $\ell\in\{1,\ldots,n\}$ (\cref{cond:indegbound}), we obtain from \cref{cond:balance}
	\begin{align*} 
		s = \sum_{i=1}^n x_i &= \sum_{v_i^- \in X} x_i +  \sum_{v_i^- \in V_N^- \setminus X} x_i
		\overset{\eqref{eq:X-has-little-flow}}{\le} (\degpara)^2+\sum_{v_i^- \in V_N^- \setminus X} \degpara\\
                                     &= (\degpara)^2 + |V_N^-\setminus X|\cdot \degpara= (\degpara)^2 + n-|X|\cdot \degpara\\
          &\le (\degpara)^2 + \degpara \cdot \degpara.
	\end{align*}
	This contradicts $s > 2(\degpara)^2$ (\cref{cond:minsize}) and hence proves the claim.

	Now, let $f$ be a maximum flow in~$N$ (computable in $O(|V_N||E_N|)=O(n(n^2-m))$ time~\cite{Orlin13}) and let $A':=\{(v_i,v_j) \in V^2 \mid f((v_i^+,v_j^-))=1\}$. Note that $|A'|=s$ and~$A'\cap A =\emptyset$.
	Clearly, for the digraph $D':=D+A'$ it holds~$\deg_{D'}(v_i)=\deg_D(v_i)+(x_i,y_i)$ for all~$v_i\in V$.\qed
\end{proof}

We remark that similar flow-constructions as given in the proof above have been used before~\cite{Gal57,CMPPS14}.
The difference here is that we actually argue about the size of the flow and not only about polynomial-time solvability. Consequently, our proof uses quite different arguments. 
% We remark that a similar flow-construction as given in the above proof was already used by \citet{Gal57} to prove the Gale-Ryser Theorem which characterizes the pairs of integer sequences that can be realized as a bipartite graph.
% However, since we consider another problem and the two results are incomparable, our proof requires different arguments. 
% Furthermore, our construction generalizes the one of \citeauthor{Gal57} by also encoding the input graph.

With \Cref{lem:factor} we have the key which allows us to transfer solutions of \nDDCSC to solutions of \DDCSC. The following lemma is immediate.

\begin{lemma}\label[lemma]{cor:LargePiTupleSolution}
  Let $I:=(D=(V,A),s,\tau,\Pi)$ with $V=\{v_1,\dots,v_n\}$ be an instance of \DDCSC with $s>2(\degpara)^2$.
  If there exists an $s'$ with $2(\degpara)^2 < s' \le s$ such that
  $I':=(\deg_D(v_1),\ldots,\deg_D(v_n),s',\tau',\Pi)$ with~$\tau'(i):=\tau(v_i)$ for all $v_i\in V$ is a yes-instance of \nDDCSC, then also~$I$ is a yes-instance of \DDCSC.
\end{lemma}

We now have all ingredients for our first main result, namely transferring
fixed-parameter tractability with respect to the combined parameter~$(s,\degpara)$ to
fixed-parameter tractability with respect to the single parameter~$\degpara$, provided that \nDDCSC is fixed-parameter tractable with respect to the largest possible integer~$\xi$ in the output sequence.
The idea is to search for large solutions based on \Cref{cor:LargePiTupleSolution} using \nDDCSC.
If there are no large solutions (that is, $s \le 2(\degpara)^2$), then we run an FPT-algorithm with respect to~$(s,\degpara)$.

\begin{theorem}\label{thm:FPT_transfer}
  If \nDDCSC is fixed-parameter tractable with respect to the largest possible integer~$\xi$ in the output sequence and \DDCSC is fixed-parameter tractable with respect to $(s,\degpara)$, then \DDCSC is fixed-parameter tractable with respect to~$\degpara$.
\end{theorem}

\begin{proof}
  In the following, let $\mathcal{A}$ be the fixed-parameter algorithm solving \DDCSC in $h(s,\degpara)\cdot n^{O(1)}$ time
  and let $\mathcal{A}'$ be the fixed-parameter algorithm solving \nDDCSC in $h'(\xi)\cdot n^{O(1)}$ time.
  Let $I:=(D=(V,A),s,\tau,\Pi)$ be a \DDCSC instance.

  If $s\le 2(\degpara)^2$, then we can run algorithm~$\mathcal{A}$ on~$I$
  in $h(s,\degpara)\cdot n^{O(1)} \le g(\degpara)\cdot n^{O(1)}$ time for some function~$g$.

  Otherwise, we check for each~$s'\in\{2(\degpara)^2+1,\ldots,s\}$ whether the instance $I_{s'}:=(\deg_D(v_1),\ldots,\deg_D(v_n),s',\tau',\Pi)$ with $\tau'(i):=\tau(v_i)$ for all~$v_i\in V$ is a yes-instance of \nDDCSC using algorithm~$\mathcal{A}'$. Note that the running time is at most $s\cdot h'(\degpara)\cdot n^{O(1)}$.
  If we find a yes-instance~$I_{s'}$ for some~$s'$, then we know by \cref{cor:LargePiTupleSolution} that $I$ is also a yes-instance. If~$I_{s'}$ is a no-instance for all~$s'\in\{2(\degpara)^2+1,\ldots,s\}$, then we also know that there cannot exist a solution for~$I$ of size larger than~$2(\degpara)^2$ since the existence of a solution for a \DDCSC instance clearly implies a solution for the corresponding \nDDCSC instance.
  Therefore, $I$ is a yes-instance if and only if $I':=(D,2(\degpara)^2,\tau,\Pi)$ is a yes-instance.
  We can thus run algorithm~$\mathcal{A}$ on~$I'$ in~$h(2(\degpara)^2,\degpara)\cdot n^{O(1)}$ time.\qed
\end{proof}

Our second main result allows us to transfer a polynomial-size problem kernel with respect to~$(s,\degpara)$
to a polynomial-size problem kernel with respect to~$\degpara$ if \nDDCSC is polynomial-time solvable.
The proof is analogous to the proof of \cref{thm:FPT_transfer}.

\begin{theorem}\label{thm:Kernel_transfer}
  If \DDCSC admits a problem kernel containing $g(s,\degpara)$ vertices computable in $p(n)$ time and \nDDCSC is solvable in $q(n)$ time for polynomials~$p$ and~$q$, then \DDCSC admits a problem kernel with $g(2(\degpara)^2,\degpara)$ vertices computable in~$O(s\cdot q(n) + p(n))$ time.
\end{theorem}

\begin{proof}
  Let $I:=(D=(V,A),s,\tau,\Pi)$ be a \DDCSC instance.
  If $s\le 2(\degpara)^2$, then we simply run the kernelization algorithm on~$I$
  obtaining an equivalent instance containing at most~$g(2(\degpara)^2,\degpara)$ vertices in~$p(n)$ time.
  Otherwise, we check, in~$q(n)$ time for each~$s'\in\{2(\degpara)^2+1,\ldots,s\}$, whether the instance $I_{s'}:=(\deg_D(v_1),\ldots,\deg_D(v_n),s',\tau',\Pi)$ with $\tau'(i):=\tau(v_i)$ for all~$v_i\in V$ is a yes-instance of \nDDCSC. Note that the running time is thus at most~$s\cdot q(n)$.
  If we find a yes-instance~$I_{s'}$ for some~$s'$, then we know by \cref{cor:LargePiTupleSolution} that also~$I$ is a yes-instance, and thus we return a trivial \DDCSC yes-instance. If~$I_{s'}$ is a no-instance for all~$s'\in\{2(\degpara)^2+1,\ldots,s\}$, then we also know that there cannot exist a solution for~$I$ of size larger than~$2(\degpara)^2$ since the existence of a solution for a \DDCSC instance clearly implies a solution for the corresponding \nDDCSC instance.
  Therefore, $I$ is a yes-instance if and only if $I':=(D,2(\degpara)^2,\tau,\Pi)$ is a yes-instance. Again, we run the kernelization algorithm on~$I'$ and return an equivalent instance with at most~$g(2(\degpara)^2,\degpara)$ vertices in~$p(n)$ time.
  The overall running time is thus in~$O(s \cdot q(n)+p(n))$ and we obtain a problem kernel with respect to~$\degpara$.\qed
\end{proof}

\section{Applications}

In the following, we show how the framework described in \cref{sec:GeneralSetting} can be applied to three special cases of \DDCSC{}.
These special cases naturally extend known problems on undirected graphs to the digraph setting.

\subsection{Digraph Degree Constraint Completion} \label{sec:DDCC}

In this section, we investigate the NP-hard special case of \DDCSC\footnote{This special case was investigated more specifically
in the Bachelor thesis of \citeauthor{Kos15}~\cite{Kos15} (available online).} where the property~$\Pi$ allows any possible degree sequence, see \Cref{fig:flownw18} for two illustrating examples.

\begin{problemdef}
  \problemtitle{\DDCClong (\DDCC)}
    \probleminput{A digraph $D=(V,A)$, a positive integer $s$, and a ``degree list function'' $\tau\colon V \rightarrow 2^{\{0,\dots,r\}^2}$.}
    \problemquestion{Is it possible to obtain a digraph $D'$ by inserting at most $s$ arcs in~$D$ such that $\deg_{D'}(v) \in \tau(v)$ for all $v \in V$?}
\end{problemdef}

\begin{figure}[t]
\centering
\begin{tikzpicture}[>=stealth, scale=0.9]

\node[gnode] (s-1) at (0,0) {};
\node[gnode] (s-2) at (2,0) {};
\node[gnode] (s-3) at (4,0) {};

\node at (0,0.5) {$\{(0,1)\}$};
\node at (2,0.5) {$\{(1,0),(2,0)\}$};
\node at (4,0.5) {$\{(0,1)\}$};

\node[gnode] (s-4) at (7,0) {};
\node[gnode] (s-5) at (9,0) {};
\node[gnode] (s-6) at (11,0) {};

\node at (7,0.5) {$\{(0,1)\}$};
\node at (9,0.5) {$\{(2,0)\}$};
\node at (11,0.5) {$\{(1,1),(2,1)\}$};

\draw[->] (s-1)  -- (s-2);
\draw[->, dashed] (s-3) -- (s-2);
\draw[->] (s-4) -- (s-5);
    
\end{tikzpicture}
\captionof{figure}{Two example instances of \DDCC with $s=1$.
The left instance is solvable by inserting the (dashed) arc from the right vertex to the middle vertex.
The right instance is a no-instance since one cannot add an outgoing arc to the left vertex or to the middle vertex but one has to add an incoming arc to the right vertex (loops are not allowed).}
\label{fig:flownw18}
\end{figure}
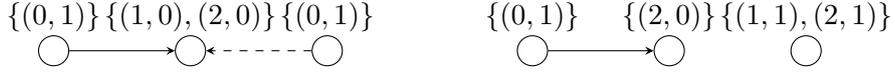
\DDCC is the directed (completion) version of the well-studied undirected \textsc{Degree Constraint Editing} problem~\cite{MS12, Gol15} for which an~$O(r^5)$-vertex problem kernel is known~\cite{FNN16}.
We subsequently transfer the polynomial-size problem kernel for the undirected case to a polynomial-size problem kernel for \DDCC with respect to~\degpara.
Note that the parameter~\degpara is clearly at most~$r$.
Since it is trivial to decide~$\Pi$ in this case, due to \cref{thm:PiEAfptDeltak} we obtain fixed-parameter tractability of \DDCC with respect to~$(s,\Delta_D)$.
The result is based on a bounded search space, namely a $2s(\Delta_D+1)$-type set (see \cref{def:alpha-block-type-set} and \cref{lem:DDCC-type-set}).
We further strengthen this result by removing all vertices that are not in the $2s(\Delta_D+1)$-type set and adjusting the degree list function~$\tau$ appropriately.
\cref{lem:DDCC-type-set} then yields the correctness of this approach resulting in a polynomial-size problem kernel with respect to~$(s,\degpara)$.

We start with the following simple reduction rule.
Recall that a vertex~$v$ is called unsatisfied if~$\deg_D(v)\not\in\tau(v)$.
  
\begin{rrule}\label{rr:DDCCtrivialNO}
  Let $(D=(V,A),s,\tau)$ be a \DDCC instance.
  If there are more than~$2s$ unsatisfied vertices, then return a trivial no-instance.
  Moreover, if there exists a vertex~$v\in V$ with $\deg_D^-(v) > \degpara$ or $\deg_D^+(v) > \degpara$, then also return a trivial no-instance.
\end{rrule}

\begin{lemma}\label[lemma]{lem:DDCCtrivialNO}
\cref{rr:DDCCtrivialNO} is correct and can be computed in $O(m+|\tau|)$ time.  
\end{lemma}

\begin{proof}
  If there are more than~$2s$ unsatisfied vertices, then we can return a trivial no-instance since
  inserting an arc can satisfy at most two vertices.
  Also, by inserting arcs we can only increase in- and outdegrees of vertices. Hence, we can return a no-instance if the in- or outdegree of a vertex is larger than $\degpara$.
  This proves the correctness.
  
  The reduction rule is applicable in $O(m+|\tau|)$ time by computing the degree of each vertex in $O(n+m)$ time and subsequently iterating through~$\tau$.\qed
\end{proof}

Based on \cref{rr:DDCCtrivialNO}, we obtain a polynomial-size
problem kernel with respect to the combined parameter~$(s,\degpara)$ as follows.

\begin{theorem}\label[theorem]{thm:DDCCkrKernel}
\DDCC admits a problem kernel where the number of vertices is in $ O(s(\degpara)^3) \subseteq O(sr^3)$. 
It is computable in $O(m+|\tau|+r^2)$ time.
\end{theorem}

\begin{proof}
  Let $I = (D=(V,A),s,\tau)$ be an instance of \DDCC.
  First, we apply \cref{rr:DDCCtrivialNO} in $O(m+|\tau|)$ time. If a no-instance is returned, then we are done. Otherwise, we know that there are at most~$2s$ unsatisfied vertices.
  Also, we know that $\deg_D^-(v)\le \degpara$ and $\deg_D^+(v)\le \degpara$ for all~$v\in V$.
  We compute an $\alpha$-type set~$C$ (see \cref{def:alpha-block-type-set}) for $\alpha:=2s(\Delta_D+1)$ in $O(m+|\tau|+r^2)$ time (\cref{lem:alphaSetLinTime}) and return the instance $I' = (D':=D[C],s,\tau_C)$, where the adjusted degree list~$\tau_C(v)$, for each~$v\in C$, is defined as follows:
	\[
    \tau_C(v) \coloneqq \{(i,j)\in\{0,\ldots,\degpara\}^2 \mid (i,j)+(|N_D^-(v)\setminus C|,|N_D^+(v)\setminus C|)\in\tau(v)\}.
	\]
    The instance~$I'$ can be computed in~$O(m+|\tau|+r^2)$ time.
    We now show that~$I'$ is an equivalent instance of \DDCC.

    Assume that~$I'$ is a yes-instance, that is, there exists a set~$A'\subseteq C^2$ of size at most~$s$ such that $\deg_{D'+A'}(v)\in\tau_C(v)$ for each $v\in C$.
    Then, the set~$A'$ is also a solution for~$I$ since, for each vertex~$v\in C$, \[\deg_{D+A'}(v) =\deg_{D'+A'}(v)+(|N_D^-(v)\setminus C|,|N_D^+(v)\setminus C|)\in\tau(v)\] by definition of~$\tau_C(v)$.
    Moreover, for each vertex~$v\in V\setminus C$, we know that $\deg_{D+A'}(v)=\deg_{D}(v)\in\tau(v)$ since $V\setminus C$ contains only satisfied vertices. Hence, $I$ is a yes-instance.

    Conversely, let~$I$ be a yes-instance. Then, by \cref{lem:DDCC-type-set}, we know that there exists an arc set~$A^*\subseteq C^2$ of size at most~$s$ such that $\deg_{D+A^*}(v)\in\tau(v)$ for all~$v\in V$.
    Then, for each vertex~$v\in C$, \[\deg_{D'+A^*}(v) = \deg_{D+A^*}(v)-(|N_D^-(v)\setminus C|,|N_D^+(v)\setminus C|)\in\tau_C(v)\] by definition of~$\tau_C$.
    Hence, also~$I'$ is a yes-instance.
    
    Concerning the size of~$D'$, observe that~$C$ contains at most~$2s$ unsatisfied vertices and at most~$\alpha$ satisfied vertices for each of the~$(\degpara)^2$ possible types.
    Therefore,
    \[
      |C| \le 2s  + (\degpara+1)^2\cdot \alpha \le 2s + (\degpara)^2\cdot 2s(\Delta_D+1).
    \]
Since $\Delta_D\le \degpara$, we obtain a problem kernel with~$O(s(\degpara)^3)$ vertices.
  The overall running time is in~$O(m+|\tau|+r^2)$.\qed
\end{proof}

The goal now is to use our framework (\cref{thm:Kernel_transfer}) to transfer the polyno\-mial-size problem kernel with respect to~$(s,\degpara)$ to a polynomial-size problem kernel with respect to~$\degpara$ alone.
To this end, we show that the corresponding number problem \nDDCC (which is the special case of \nDDCSC without the sequence property~$\Pi$) is polynomial-time solvable.

\begin{problemdef}
    \problemtitle{\textsc{\nDDCClong} (\nDDCC)}
    \probleminput{A sequence $\sigma = (c_1,d_1),\ldots,(c_n,d_n)$ of~$n$ nonnegative integer tuples,
                  a positive integer~$s$, and
                  a ``tuple list function'' $\tau\colon \{1,\ldots,n\} \rightarrow 2^{\{0,\ldots,r\}^2}$.}
    \problemquestion{Is there a sequence $\sigma' = (c'_1,d'_1),\ldots,(c'_n,d'_n)$
                     such that $\sum_{i=1}^n (c'_i-c_i) = \sum_{i=1}^n (d'_i-d_i) = s$, and
                     $c_i \le c'_i$, $d_i \le d'_i$, and $(c'_i, d'_i)\in\tau(i)$ for all $1 \le i \le n$?}
\end{problemdef}

\nDDCC can be solved in pseudo-polynomial time by a dynamic programming algorithm. 
Note that pseudo-polynomial time is sufficient for our purposes since all occurring numbers will be bounded by~$O(n^2)$ when creating the \nDDCC instance from the given \DDCC instance.
(In fact, we conjecture that \nDDCC is weakly NP-hard and a reduction from \textsc{Partition} should be possible as in the case for \nDA in \cref{sec:DA}, \cref{thm:nDAhard}.)

\begin{lemma}\label[lemma]{lem:nDDCC-poly}
\nDDCC is solvable in $O(n(sr)^2)$ time.
\end{lemma}

\begin{proof}
Let $I:=((c_1,d_1),\ldots,(c_n,d_n),s,\tau)$ be an instance of \nDDCC.
We solve~$I$ using a modified version of the dynamic programming algorithm for \NCEA due to \citet[Lemma~2]{FNN16}.
To this end, we define the Boolean table~$M[i,j,l]$ for $i\in\{1\ldots,n\}$, $j,l\in\{0,\ldots,s\}$,
where~$M[i,j,l]=\texttt{true}$ if and only if there exist tuples~$(c_1',d_1'),\ldots,(c_i',d_i')$ with $c_p'\ge c_p$, $d_p'\ge d_p$ and $(c_p',d_p')\in\tau(p)$ for all $p\in\{1,\ldots,i\}$ such that
$\sum_{p=1}^i(c_p'-c_p)=j$ and $\sum_{p=1}^i(d_p'-d_p)=l$.
Thus, $I$ is a yes-instance if $M[n,s,s]=\texttt{true}$.
We compute~$M$ based on the recurrence where we essentially consider all possibilities to fix the~$i$-th tuple and recurse:
\begin{align*}
  &M[i,j,l] = \texttt{true} \quad \Leftrightarrow \\
  &\exists (c_i',d_i') \in \tau(i) : (c_i' \geq c_i) \wedge (d_i' \geq d_i) \wedge M[i-1,j-(c_i'-c_i),l-(d_i'-d_i)],
\end{align*}
where we set
\[
M[1,j,l] := \begin{cases} \texttt{true,} & \text{if }(c_1+j,d_1+l) \in \tau(1),
  \\ \texttt{false,} & \text{otherwise.} \end{cases}
\]
The size of $M$ is in $O(ns^2)$. A single entry can be computed in~$O(r^2)$ time.\qed
\end{proof}

Combining \cref{thm:DDCCkrKernel,lem:nDDCC-poly} yields the following corollary of \cref{thm:Kernel_transfer}.

\begin{corollary}\label[corollary]{cor:DDCCkernel}
	\DDCC admits a problem kernel containing $O((\degpara)^5) \subseteq O(r^5)$ vertices.
	It is computable in~$O(m + ns^3r^2)$ time.
\end{corollary}

\subsection{Digraph Degree Sequence Completion}
\label{sec:DDSC}

\begin{figure}[t]
  \centering
  \begin{tikzpicture}
    \node[gnode] (v1) at (0,0) {};
    \node[gnode] (v2) at (2,0) {};
    \node[gnode] (v3) at (4,0) {};
    \node[gnode] (v4) at (2,0.7) {};
    \node[right] at (4.5,0.4) {$\sigma = \{(0,3),(1,1),(2,0),(2,1)\}$};
    
    \draw[->] (v1) to (v2);
    \draw[->] (v2) to (v3);
    \draw[->] (v4) to (v2);
    \draw[->] (v4) to (v3);
    \draw[->, dashed] (v4) to (v1);
  \end{tikzpicture}
  \caption{Example instance of \DDSC.
           Inserting the dashed arc in the input digraph (solid arcs) with degree sequence~$\{(0,1),(0,2),(2,0),(2,1)\}$ yields a digraph with the given target sequence~$\sigma$.}
  \label{DDSCexample}
\end{figure}
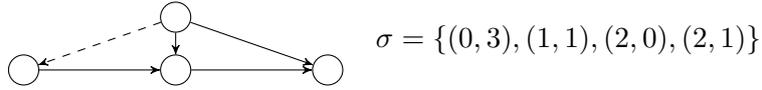

In this section, we investigate the NP-hard special case of \DDCSC\footnote{Although not stated explicitly, the NP-hardness follows from the proof of Theorem 3.2 of the Bachelor thesis of \citet{Mil15} (available online) as the construction therein allows for only one feasible target degree sequence.}
 where~$\tau$ does not restrict the allowed degree of any vertex and $\Pi$~is fulfilled by exactly one specific degree sequence~$\sigma$ (see \cref{DDSCexample} for an example).
The undirected problem variant is studied by \citet{GM17}.
\begin{problemdef}
	\problemtitle{\DDSClong (\DDSC)}
	\probleminput{A digraph $D = (V,A)$, a digraph degree sequence~$\sigma$ containing~$|V|$ integer tuples.}
	\problemquestion{Is it possible to obtain a digraph~$D'$ by inserting arcs in~$D$ such that~$\sigma(D')=\sigma$?}
\end{problemdef}

For \DDSC, the parameter~\degpara is by definition equal to~$\Delta_\sigma$.
Moreover, note that the number~$s$ of arcs to insert (if possible) is determined by the target sequence~$\sigma$ by~$s := \sum_{(c,d)\in\sigma}c - \sum_{v\in V(D)}\deg_D^-(v)$.
We henceforth assume that \[s = \sum_{(c,d)\in\sigma}c - \sum_{v\in V(D)}\deg_D^-(v) = \sum_{(c,d)\in\sigma}d - \sum_{v\in V(D)}\deg_D^+(v) \ge 0\]
holds since otherwise we have a trivial no-instance.

Since deciding~$\Pi$ (that is, deciding whether $\sigma(D')=\sigma$) can be done in polynomial time, we immediately obtain fixed-parameter tractability of \DDSC with respect to~$(s,\Delta_D)$ due to \cref{thm:PiEAfptDeltak}.
We further strengthen this result by developing a polynomial-size problem kernel for \DDSC with respect to~$(s,\Delta_\sigma)$.
The kernelization is inspired by the~$O(s\Delta_\sigma^2)$-vertex problem kernel for the undirected problem by~\citet{GM17}.
The main idea is to only keep the vertices of a $2s(\Delta_D+1)$-block set (see \cref{def:alpha-block-type-set}) together with some additional ``dummy'' vertices and to adjust the digraph degree sequence~$\sigma$ properly.

\begin{theorem}
  \label{thm:DDSC_kDelta-kernel}
  \DDSC admits a problem kernel containing~$O(s\Delta_\sigma^3)$ vertices computable in~$O(n+m+\Delta_\sigma^2)$ time.
\end{theorem}

\begin{proof}
  Let~$(D,\sigma)$ be a \DDSC instance. % First, we apply \cref{rr:DDSC_trivialNO} in~$O(n\log n+m)$ time. If a no-instance is returned, then we are done.
  Clearly, since we are only allowed to insert arcs in the digraph~$D$, we can never decrease the in- or outdegree of any vertex. Hence, if~$\Delta^-_D> \Delta^-_\sigma$ or $\Delta^+_D > \Delta^+_\sigma$, then we return a trivial no-instance.
  Otherwise, we know that~$\Delta_D\le\Delta_\sigma$.
  Moreover, since inserting one arc can change the degrees of at most two vertices, it also holds $\lambda_D(\deg_D(v))\le\lambda_\sigma(\deg_D(v))+2s$ for each~$v\in V(D)$.

  We now compute a $2s(\Delta_D+1)$-block set~$C$ (see \cref{def:alpha-block-type-set}) in~$O(n+m+\Delta_D^2)$ time (\cref{lem:alphaSetLinTime}) and return the instance~$(D',\sigma')$ which is defined as follows.
  The digraph~$D'$ is constructed from~$D$ by the following steps:
  \begin{itemize}
    \item Delete all vertices of~$V(D)\setminus C$.
    \item Add $h:=\Delta_\sigma+2$ new vertices~$W:=\{w_1,\ldots,w_h\}$ and insert all arcs~$W^2$.
    \item For each~$v\in C$ such that the number~$r^-_v:=|N_D^-(v)\setminus C|$ of inneighbors in~$V(D)\setminus C$ is at least one, insert the arcs~$\{(w_i,v)\mid 1\le i \le r^-_v\}$.
      \item For each~$v\in C$ such that the number~$r^+_v:=|N_D^+(v)\setminus C|$ of outneighbors in~$V(D)\setminus C$ is at least one, insert the arcs~$\{(v,w_i)\mid 1\le i \le r^+_v\}$.
      \end{itemize}
      The digraph~$D'$ can be constructed in~$O(n+m+\Delta_\sigma^2)$ time.
  Observe that $\deg_{D'}^-(w_i)\ge \Delta_\sigma+1$ and $\deg_{D'}^+(w_i)\ge\Delta_\sigma+1$ hold for all~$i\in\{1,\ldots,h\}$, and that~$\deg_{D'}(v)=\deg_D(v)\le \Delta_\sigma$ holds for all~$v\in C$.
  The number of vertices in~$D'$ equals~$|C|+h$. Note that~$C$ contains at most~$2s(\Delta_D+1)$ vertices of each of the~$(\Delta_D+1)^2$ possible vertex degrees in~$D$. Thus, $D'$ contains $O(s\Delta_\sigma^3)$ vertices.
  
  The digraph degree sequence~$\sigma'$ is constructed from~$\sigma$ as follows:
  \begin{itemize}
    \item For each vertex~$v\in V(D)\setminus C$ that was removed from~$D$, remove a copy of the tuple~$\deg_D(v)$ from~$\sigma.$
    \item For each $i\in\{1,\ldots,h\}$, add the tuple $\deg_{D'}(w_i)$.
    \end{itemize}
    Note that this construction is well-defined, that is, we can always apply the first step and remove a copy of~$\deg_D(v)$ from~$\sigma$ since we remove at most \[\lambda_D(\deg_D(v) - 2s(\Delta_D+1) < \lambda_D(\deg_D(v))-2s\le \lambda_\sigma(\deg_D(v))\] copies.
    The construction of~$\sigma'$ can be done in~$O(n)$ time. Hence, the overall running time of computing the problem kernel is in~$O(n+m+\Delta_\sigma^2)$.

    It remains to show that~$(D',\sigma')$ is a yes-instance if and only if~$(D,\sigma)$ is a yes-instance.
    Assume first that~$(D,\sigma)$ is a yes-instance. We know from \cref{lem:PiDDSC-block-set} that there exists a solution~$A^*\subseteq C^2$ with~$\sigma(D+A^*)=\sigma$.
    Using
    \begin{align*}
      \forall v\in V(D)\setminus C &:\deg_{D+A^*}(v)=\deg_{D}(v),\\
      \forall v\in C &: \deg_{D'+A^*}(v)=\deg_{D+A^*}(v), \text{ and}\\
      \forall  w_i\in W &:\deg_{D'+A^*}(w_i)=\deg_{D'}(w_i),
    \end{align*}
it is then easy to verify that~$\sigma(D'+A^*)=\sigma'$, and thus, $(D',\sigma',s)$ is a yes-instance.

Conversely, let~$A'\subseteq V(D')^2$ be a solution for~$(D',\sigma')$
with~$\sigma(D'+A')=\sigma'$.
We claim that~$A'\subseteq C^2$, that is,~$A'$ does not contain an arc incident to a vertex in~$W$. To see this, recall that by construction
\begin{align*}
  \deg^-_{D'}(w_1) &=\Delta^-_{\sigma'} \ge \ldots \ge \deg^-_{D'}(w_h) \ge \Delta_\sigma + 1 > \deg^-_{D'}(v) \text{ and}\\
  \deg^+_{D'}(w_1) &=\Delta^+_{\sigma'} \ge \ldots \ge \deg^+_{D'}(w_h) \ge \Delta_\sigma + 1 > \deg^+_{D'}(v)
\end{align*}
hold for all~$v\in C$.
That is, $\deg_{D'}(w_1)=(\Delta_{\sigma'}^-,\Delta_{\sigma'}^+)$, and thus a solution must not insert arcs incident to~$w_1$. It follows that $\deg_{D'+A'}(w_1)=\deg_{D'}(w_1)$.
This recursively also holds for $w_2,\ldots,w_h$ and thus, we have $\deg_{D'+A'}(w_i)=\deg_{D'}(w_i)$ for all~$w_i\in W$.
Hence, $A'$ does not contain any arcs incident to vertices in~$W$, that is, $A'\subseteq C^2$.
Thus, we can derive
\begin{align*}
  \forall v\in C &: \deg_{D'+A'}(v)=\deg_{D+A'}(v), \text{ and}\\
  \forall v\in V(D)\setminus C &:\deg_{D+A'}(v)=\deg_{D}(v).
\end{align*}
It is now straightforward to check that $\sigma(D+A')=\sigma$.\qed
\end{proof}
To apply our framework and derive a polynomial-size problem kernel with respect to~$\Delta_\sigma$ for \DDSC, we define a corresponding number problem and show its polynomial-time solvability.
The number problem \nDDSC is the special case of \nDDCSC asking for the specific target sequence~$\sigma$.

\begin{problemdef}
    \problemtitle{\textsc{\nDDSClong (\nDDSC)}}
    \probleminput{Two multisets $\sigma=\{(c_1,d_1),\ldots,(c_n,d_n)\}$ and~$\phi=\{(c_1',d_1'),\ldots,(c_n',d_n')\}$ containing~$n$ nonnegative integer tuples.}
    \problemquestion{Is there a bijection~$\pi:\{1,\ldots,n\} \to \{1,\ldots,n\}$ such that
     $c_i \le c'_{\pi(i)}$, and $d_i \le d'_{\pi(i)}$, for all $1 \le i \le n$?}
\end{problemdef}

\nDDSC can be solved in polynomial time by finding perfect matchings in an auxiliary graph.

\begin{lemma}
  \label[lemma]{lem:TSC-poly}
  \nDDSC is solvable in~$O(n^{2.5})$ time.
\end{lemma}

  \begin{proof}
    We show how to solve the problem by computing a perfect matching in a bipartite graph.
    Let $(\{(c_1,d_1),\ldots,(c_n,d_n)\},\{(c_1',d_1'),\ldots,(c_n',d_n')\})$ be a \nDDSC instance. 
    % To start with, note that the first condition in the problem definition is independent of the respective ordering since, for any ordering~$\pi_1,\ldots,\pi_n$ of~$1,\ldots,n$, it holds
    % \begin{align*}
    %   \sum_{i=1}^n(c_{\pi_i}'-c_i) = \sum_{i=1}^nc_{\pi_i}'-\sum_{i=1}^nc_i
    % \end{align*}
    % If the above equation does not hold, then we reject the instance.
    % The same holds analogously for the sum over the $d_i$'s. 
    % Thus, we can first check whether
    % \[\sum_{i=1}^nc_i'-\sum_{i=1}^nc_i = \sum_{i=1}^nd_i'-\sum_{i=1}^nd_i=s\]
    % holds and otherwise reject the instance. 
    % If the above condition is met, then we try to find an ordering of the tuples in~$\phi$ as follows:
    We construct an undirected bipartite graph~$G:=(V\cup W,E)$. For each~$i\in\{1,\ldots,n\}$, there is a vertex~$v_i\in V$ corresponding to the tuple~$(c_i,d_i)$, and a vertex~$w_i\in W$ corresponding to~$(c_i',d_i')$.
    For each~$i,j\in\{1,\ldots,n\}$, $i\neq j$, the edge~$\{v_i,w_j\}$ is in~$E$ if and only if~$c_j'\ge c_i$ and~$d_j'\ge d_i$ hold.
    The graph~$G$ can be computed in~$O(n^2)$ time. Note that a perfect matching in~$G$
    defines a bijection that satisfies the condition in the problem definition. Hence, we can solve a \nDDSC instance by computing a perfect matching in a bipartite graph, which can be done in~$O(|E|\sqrt{|V\cup W|})=O(n^{2.5})$ time~\cite{HK73}.\qed
  \end{proof}

Combining \cref{thm:DDSC_kDelta-kernel} and \cref{lem:TSC-poly} yields the following corollary of \cref{thm:Kernel_transfer}.

\begin{corollary}\label[corollary]{cor:DDSCkernel}
	\DDSC admits a problem kernel containing $O(\Delta_\sigma^5)$ vertices. 
	It is computable in $O(sn^{2.5})$ time.
\end{corollary}

\subsection{Degree Anonymity}
\label{sec:DA}

We extend the definition of \textsc{Degree Anonymity} in undirected graphs due to \citet{LT08} to digraphs and obtain the following NP-hard problem~\cite{Mil15} (\cref{example of DA} presents an example): 

\begin{problemdef}
	\problemtitle{\DAlong (\DA)}
	\probleminput{A digraph $D = (V,A)$ and two positive integers $k$ and $s$.}
	\problemquestion{Is it possible to obtain a digraph $D'$  by inserting at most~$s$ arcs in~$D$ such that~$D'$ is $k$-anonymous, that is, for every vertex $v \in V$ there are at least $k-1$~other vertices in $D'$ with degree~$\deg_{D'}(v)$?}
\end{problemdef}

\begin{figure}[t]
	\centering
	\begin{tikzpicture}
			\node[gnode] (a1) at (0,0.25) {};
			\node[gnode] (a2) at (1.5,0.25) {};
			\node[gnode] (b1) at (4,0.25) {};
			\node[gnode] (b2) at (5.5,0.25) {};

                        \draw[->, bend angle=30, bend left] (a1) to (a2);
                        \draw[->, bend angle=30, bend left] (a2) to (a1);
                        \draw[->, bend angle=30, bend left] (b1) to (b2);
                        \draw[->, bend angle=30, bend left] (b2) to (b1);
                        
			\node[gnode] (c1) at (8,0) {};
			\node[gnode] (c2) at (9,0.5) {};
			\node[gnode] (c3) at (10,0) {};

                        \draw[->] (c2) to (c1);
                        \draw[->] (c3) to (c2);
                        \draw[->, dashed] (c1) to (c3);
	\end{tikzpicture}
      \caption{Example instance of \DA. The input digraph with three components (solid arcs) is $1$-anonymous since there is only one vertex with degree $(0,1)$. By inserting the dashed arc, the digraph becomes $7$-anonymous since all vertices have degree~$(1,1)$.}
	\label{example of DA}
\end{figure}
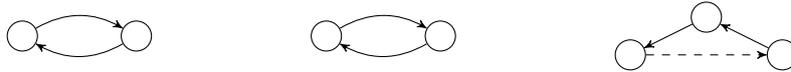

The (parameterized) complexity as well as the (in-)approx\-imability of the undirected version called \textsc{Degree Anonymity} are well-studied~\cite{CKSV13,HNNS15,BBHNW16}.
There also exist many heuristic approaches to solve the undirected version~\cite{CHT17,HHN14}. % LSB12,
Notably, our generic approach shown in~\cref{sec:GeneralFPTr} originates from a heuristic of~\citet{LT08} for \textsc{Degree Anonymity}.
Later, \citet{HNNS15} used this heuristic to prove that ``large'' solutions of \textsc{Degree Anonymity} can be found in polynomial time and \citet{FNN16} extended this approach to a more general class of problems.
The property~$\Pi$ (that is, $k$-anonymity) can clearly be checked for a given input digraph degree sequence in polynomial time. Hence, \cref{thm:PiEAfptDeltak} yields fixed-parameter tractability of~\DA with respect to~$(s,\Delta_D)$.
Again, we develop a polynomial-size problem kernel with respect to~$(s,\Delta_D)$.
Somewhat surprisingly, we cannot transfer this problem kernel to a problem kernel with respect to~$\degpara$ since we are not able to solve the corresponding number problem in polynomial time.
In fact, we will show that it is at least weakly NP-hard.
% Surprisingly, while we can apply our generic framework for the problems shown in the previous two subsections, \DA seems to be more intricate.
% As we will subsequently indicate, the source of the difficulties seems to lie within the number problem.

To start with, we give a problem kernel based on \Cref{lem:PiDDSC-block-set} in a similar fashion as in the proof of \Cref{thm:DDSC_kDelta-kernel}.
More precisely, by \Cref{lem:PiDDSC-block-set} we know that we only need to keep a $2s(\Delta_D+1)$-block set~$C$, that is, $2s(\Delta_D+1)$ arbitrary vertices of each block.
Note that deleting all vertices that are not in~$C$ changes the degrees of the vertices in~$C$.
We repair this in a similar way as in the problem kernel stated in \Cref{thm:DDSC_kDelta-kernel}: 
After deleting the vertices that are not in~$C$, we add vertices adjacent to the vertices in~$C$ in such a way that the vertices in~$C$ keep their original degrees. 
Denoting the set of newly added vertices by~$P$, we also need to ``separate'' the vertices in~$P$ from the vertices in~$C$ so that their degrees do not interfere in the target degree sequence. 
We do this, similarly as in the proof of \Cref{thm:DDSC_kDelta-kernel}, by increasing the degrees of all vertices in~$P$ to at least~$\Delta_D + s + 1$. 
Furthermore, we need to ensure that a solution in the new instance does not insert arcs between vertices in~$C$ and vertices in~$P$ since we cannot map such solutions back to solutions for the original instance.
Solving this issue, however, is not as simple as for \DDSC
and requires some adjustment of the actual number of vertices we keep. 
As a result, we will prove that if there is a solution inserting arcs between~$C$ and~$P$, then there is also a solution not inserting such arcs (\Cref{lem:P-uneff}).

Another adjustment concerns the anonymity level~$k$: 
If~$k$ is large, then we need to shrink it since otherwise we would always create no-instances. 
The general idea is to keep the ``distance to size~$k$'', meaning that if in the original instance some block contains~$k+x$ vertices for some~$x \in \{-2s, \ldots, 2s\}$, then in the new instance this block should contain~$k'+x$ vertices where~$k'$ is the new anonymity level.
The reason for the specific range of values for~$x$ between~$-2s$ and~$2s$ is that if some block has size larger than~$k + 2s$ for example, then, after inserting~$s$ arcs, this block will still be of size larger than~$k$.
Similarly, if a block contains less than~$k - 2s$ vertices, then after inserting~$s$ arcs it will contain less than~$k$ vertices and it will violate the $k$-anonymity constraint unless it is empty. 
Hence, the interesting cases for~$x$ are between~$-2s$ and~$2s$.
In order to ensure that there is a solution not inserting arcs between~$C$ and~$P$, we need to increase this range from~$-2s$ to~$4s$, see the proof of \Cref{lem:P-uneff} for further details.

In the following, we describe the details of our kernelization algorithm, see \cref{alg:polyKernelDirDegAnon} for the pseudocode. 
Observe that our general approach is a non-obvious adaption of the polynomial-size problem kernel for the undirected \textsc{Degree Anonymity} problem by \citet{HNNS15}. 
\begin{algorithm}[t!]\small
  \caption{The pseudocode of the algorithm computing a polynomial-size kernel with respect to $(s,\Delta_D)$ for \DA.}
 \label{alg:polyKernelDirDegAnon}
	\KwIn{A digraph~$D = (V,A)$ and integers $k,s \in \N$.}
	\KwOut{A digraph~$D'$ and integers $k',s\in\N$.}
   \vspace{1em}
	\If(\tcp*[f]{$\beta$ is defined as~$\beta := (\Delta_D+2)2s$} \label{line:iniBound}){$|V| \le (\Delta_D+1)^2(\beta+2s)$}
	{
		\KwRet{$(D,k,s)$}
	}
	$k' \gets \min\{k,\beta\}$\; \label{line:setNewK} 
        $C \gets \emptyset$\;
	\ForEach(\label{line:keepVerticesStart}){distinct tuple $t$ occurring in $\sigma(D)$}
	{
		\If{$2s <|B_D(t)| < k - 2s$}
		{
			\KwRet{trivial no-instance} \tcp*{insufficient budget for $B_D(t)$} \label{line:trivialNoInstance}
		}
		\uIf(\tcp*[f]{determine number of retained vertices}){$k\le \beta$}
		{
			$x \gets \min \{|B_D(t)|, \beta+2s \}$ \tcp*{keep at most~$\beta+2s$ vertices} \label{line:NoOfVerticesCase1} 
		}
		\uElseIf(\tcp*[f]{``small'' block}){$|B_D(t)| \le 2s$}
		{
			$x \gets |B_D(t)|$ \tcp*{keep all vertices (``distance to size zero'')} \label{line:NoOfVerticesCase2.1} 
		}
		\Else(\tcp*[f]{``large'' block and~$k' = \beta$})
		{
			$x \gets k'+\min\{2s,(|B_D(t)|-k)\}$ \tcp*{keep ``distance to size~$k$''} \label{line:NoOfVerticesCase2.2} 
		}
		add~$x$ arbitrary vertices from~$B_D(t)$ to~$C$\; \label{line:keepVerticesEnd}
	}
	$D' \gets D[C]$\; \label{line:keepGA} % \tcp*{$D'$ stores the vertices we retain from~$G$.}
	\ForEach(\tcp*[f]{insert new vertices to preserve degrees of vertices in~$C$}){$v \in C$}
	{
	 add $\deg_D^+(v) - \deg_{D'}^+(v)$ many vertices with an incoming arc from~$v$ to~$D'$\;\label{line:add-outdeg-fix-vertices}
	add $\deg_D^-(v) - \deg_{D'}^-(v)$ many vertices with an outgoing arc to~$v$ to~$D'$\;\label{line:add-indeg-fix-vertices}
	}
	let $P_\text{in}$ be the set of vertices added in \cref{line:add-outdeg-fix-vertices} \tcp*{$\forall v \in P_\text{in}\colon \deg_{D'}(v) = (1,0)$}
	let $P_\text{out}$ be the set of vertices added in \cref{line:add-indeg-fix-vertices} \tcp*{$\forall v \in P_\text{out}\colon \deg_{D'}(v) = (0,1)$}
	\While{$\min\{|P_\text{\textnormal{in}}|,|P_\text{\textnormal{out}}|\} < \max\{\Delta_D + s + 1, k'\}$}
	{
		add a new vertex~$v$ to~$P_\text{in}$ \label{line:add-P-in-vertex}\;
		add a new vertex~$u$ to~$P_\text{out}$ \label{line:add-P-out-vertex}\;
		insert the arc~$(u,v)$ in~$D'$ \;
	}
	insert all arcs~$P_\text{in}^2$ and~$P_\text{out}^2$ in~$D'$ \label{line:high-degree-in-P} \tcp*{ensure high degree difference from vertices in~$C$}
	insert all arcs from~$P_\text{in} \times P_\text{out}$ in~$D'$  \label{line:separate-P-in-and-P-out} \tcp*{separate $P_\text{in}$ from~$P_\text{out}$}
	\KwRet{$(D',k',s)$}
\end{algorithm}

\begin{lemma}\label[lemma]{lem:P-uneff}
	Let~$(D,k,s)$ be an instance of \DA and let~$(D',k',s)$ be the instance computed by \Cref{alg:polyKernelDirDegAnon}, where~$P:=P_\text{\textnormal{in}}\cup P_\text{\textnormal{out}}$ is the set of newly added vertices.
	If there is a solution~$S \subseteq V(D')^2$ with $|S| \le s$, then there is also a solution~$S' \subseteq V(D')^2$ with $|S'|\le|S|$ such that~$V(S')\cap P=\emptyset$.
\end{lemma}

\begin{proof}
	Let~$S \subseteq V(D')^2$ be a solution for~$(D',k',s)$ such that~$V(S') \cap P \neq \emptyset$.
	We construct a new solution~$S' \subseteq V(D')^2$ such that~$|S'| \le |S|$ and~$V(S') \cap P = \emptyset$.
	The idea is to replace the endpoints of arcs that are in~$P$ by new endpoints from one ``large'' block (of size at least~$\beta+2s$, where $\beta := (\Delta_D+2)2s$) in~$C$.
	To this end, observe that if~$V(D) \le (\Delta_D+1)^2(\beta + 2s)$, then \Cref{alg:polyKernelDirDegAnon} returns the original instance (see \cref{line:iniBound}) and we are done. 
	Hence, there is at least one block~$B_{D'}(t)$ for some~$t\in\sigma(D')$ of size at least~$\beta + 2s$ since there are at most~$(\Delta_D+1)^2$ blocks.
	We will use vertices in~$B_{D'}(t)$ as a replacement for the vertices in~$P$ within the arcs of~$S$.
	
	We now construct~$S'$.
	To this end, initialize~$S' := S \cap C^2$ and insert further arcs in the following way.
	First, consider those arcs in~$S$ that have exactly one endpoint in~$P$.
	For each arc~$(u,v) \in S$ with $u \in C$ and~$v \in P$, insert the arc~$(u,w)$ in~$S'$ where~$w \in B_{D'}(t)$ such that~$w$ is not incident to any arc in~$S'$ and is not an outneighbor of~$u$.
	Since~$|B_{D'}(t)| \ge \beta + 2s = (\Delta_D+3)2s$ and~$|S'| \le s$, it follows that~$B_{D'}(t)$ contains such a vertex~$w$. 
	Similarly, for each arc~$(v,u) \in S$ with $u \in C$ and~$v \in P$, insert the arc~$(w,u)$ in~$S'$ where~$w \in B_{D'}(t)$ is a vertex not incident to any arc in~$S'$ and not an inneighbor of~$u$.
	Again, due to the size of~$B_{D'}(t)$, such a vertex exists.
	
	Second, consider those arcs in~$S$ having both endpoints in~$P$.
	For each arc~$(u,v) \in S$ with $u,v \in P$, insert the arc~$(u',v')$ in~$S'$ where~$u',v' \in B_{D'}(t)$ such that neither~$u'$ nor~$v'$ is incident to any arc in~$S'$ and~$(u',v') \notin A(D')$.
	Since~$|B_{D'}(t)| \ge \beta + 2s = (\Delta_D+3)2s$ and~$|S'| \le s$, it follows that these vertices~$u'$ and~$v'$ exist.
	Observe that after all these modifications, there are still at least~$\beta$ vertices left in~$B_{D'}(t)$.
	
	Clearly, we have~$|S'| \le |S|$.
	It remains to prove that~$D'+S'$ is $k'$-anonymous.
	To this end, observe that, since the outdegree of each vertex in~$P_\text{in}$ is at least $|P_\text{out}|-1 \ge \Delta_D + s + 1$ (see \cref{line:separate-P-in-and-P-out}) larger than the outdegree of any vertex in $P_\text{out}$, it follows that the vertices in~$P$ which are incident to an arc in~$S$ end up in blocks of~$D'+S$ that are empty in~$D'$.
	Thus, at least~$k'$ vertices in~$P$ are the head of an arc in~$S$ and at least~$k'$ vertices in~$P$ are the tail of an arc in~$S$.
	Hence, we used at least~$k'$ vertices from~$B_{D'}(t)$ as an replacement in~$S'$ and thus the blocks~$B_{D'+S'}(t + (1,0))$ and~$B_{D'+S'}(t + (0,1))$ contain at least~$k'$ vertices.
	Furthermore, all other vertices in~$C$ have the same degree in~$D'+S$ and in~$D'+S'$ and the vertices in~$P$ are not incident to any arc in~$S'$.
	Since~$S$ was a solution, it follows that also~$D'+S'$ is $k'$-anonymous.\qed
\end{proof}

We remark that parts of the proof of \cref{lem:P-uneff} are an adaption of the proof of the corresponding lemma in the undirected case~\cite[Lemma 6]{HNNS15}.

\begin{theorem}
	\label{thm:DirDegAnon-kernel-sDelta}
	\DA admits a problem kernel containing~$O( \Delta_D^5 s )$ vertices. 
	It is computable in~$O( \Delta_D^{10} s^2 + \Delta_D^3 s n )$ time.
\end{theorem}

\begin{proof}
	We use \cref{alg:polyKernelDirDegAnon} to compute the problem kernel. 
	The correctness of the kernelization follows from the following two lemmas.
        Their proofs are, however, adaptions of the corresponding undirected counterparts~\cite[Lemmas~7 and~8]{HNNS15}

	Let~$D = (V,A)$ be a digraph and~$k \in \N$. 
	An arc set~$S \subseteq V^2$ is called \emph{$k$-insertion set} for~$D$, if~$D + S$ is~$k$-anonymous.
        
	\begin{lemma}\label[lemma]{lem:polyKernelRueckrichtung-DirDegAnon}
		If the instance~$(D',k',s)$ constructed by \Cref{alg:polyKernelDirDegAnon} is a yes-instance, then $(D,k,s)$ is a yes-instance.
	\end{lemma}

	\begin{proof}[of \cref{lem:polyKernelRueckrichtung-DirDegAnon}]
		First, observe that if $k\le \beta$, then $k'=k$ and each $k$-insertion set for~$D'$ is a $k$-insertion set for~$D$ as all blocks with less than~$\beta + 2s$ vertices remain unchanged. 
		Hence, it remains to consider the case that~$k > \beta$ and thus~$k' = \beta$.
		
		Let~$S'$ be an arc set such that~$|S'| \le s$ and $D'+S'$ is $k'$-anonymous.
		By~\cref{lem:P-uneff}, we can assume that each arc in~$S'$ has both endpoints in~$C$.
		We show that~$D+S'$ is $k$-anonymous, that is, for each block~$B_{D+S'}(t)$ we have~$|B_{D+S'}(t)| \ge k$ or $|B_{D+S'}(t)| = 0$.
		To this end, we distinguish two cases on whether the corresponding block in~$D'+S'$ is empty or contains at least~$k'$ vertices.
		
		First, consider the case~$|B_{D'+S'}(t)| = 0$. 
		Since it holds that $|S'| \le s$, it follows that~$|B_{D'}(t)| \le 2s$. 
		By \cref{line:NoOfVerticesCase2.1,line:NoOfVerticesCase2.2}, it follows that~$D$ and~$D'$ contain the same vertices of degree~$t$, that is, $B_{D'}(t) = B_{D}(t)$.
		Hence, we have~$|B_{D+S'}(t)| = 0$.
		
		Second, consider the case~$|B_{D'+S'}(t)| \ge k$.
		If~$|B_{D}(t)| \ge k + 2s$, then it clearly holds that $|B_{D+S'}(t)| \ge k$ and we are done.
		Otherwise, by \cref{line:NoOfVerticesCase2.2}, we have~$|B_{D}(t)| - k = |B_{D'}(t)| - k'$.
		Since~$S'$ only contains arcs with both endpoints in~$C$, it follows that by inserting~$S$, the same vertices will be added and removed from~$B_{D}(t)$ and~$B_{D'}(t)$, that is, $|B_{D+S'}(t)| - k = |B_{D'+S'}(t)| - k'$.
		Since~$|B_{D'+S'}(t)| \ge k'$ it follows that~$|B_{D+S'}(t)| \ge k$. 
		Thus, $D+S'$ is $k$-anonymous and $(D,k,s)$ is a yes-instance.
		\hfill (Proof of \cref{lem:polyKernelRueckrichtung-DirDegAnon})
	\end{proof}

	\begin{lemma}\label[lemma]{lem:polyKernelHinrichtung-DirDegAnon}
		If $(D,k,s)$ is a yes-instance, then the instance~$(D',k',s)$ constructed by \Cref{alg:polyKernelDirDegAnon} is a yes-instance.
	\end{lemma}

	\begin{proof}
		Observe that, in the instance~$(D',k',s)$ constructed by \Cref{alg:polyKernelDirDegAnon}, for each degree~$t \in \N^2$ we have that either~$B_D(t) = B_{D'}(t)$ (in case that~$B_D(t)$ contains few vertices, see \cref{line:NoOfVerticesCase1,line:NoOfVerticesCase2.1}) or~$B_{D'}(t) \subseteq B_D(t)$ contains at least~$\beta - 2s$ vertices ($|B_{D'}(t)| \ge \beta - 2s = (\Delta_D+1)2s$, see \cref{line:trivialNoInstance,line:NoOfVerticesCase1,line:NoOfVerticesCase2.2}).
		Thus, $D'$ contains a~$(\Delta_D+1)2s$-block set.
		Since $(D,k,s)$ is a yes-instance, it follows from \cref{lem:PiDDSC-block-set} that there is a $k$-insertion set~$S$ of size at most~$s$ for~$D$ such that~$S \subseteq C^2$.
		
		We next show that~$D'+S$ is $k'$-anonymous, and hence, $(D',k',s)$ is a yes-instance.
		First, consider the case that~$k \le \beta$ and thus $k' = k$.
		Observe that every block~$B_{D'}(t)$ containing at least~$k+2s$ vertices also contains at least~$k=k'$ vertices in~$D'+S$.
		For every block~$B_{D'}(t)$ containing less than~$k+2s$ vertices it holds that~$B_{D'}(t) = B_{D}(t)$ (see \cref{line:NoOfVerticesCase1}).
		Thus, $|B_{D+S}(t)| = |B_{D'+S}(t)|$ and therefore~$B_{D'+S}(t)$ fulfills the $k$-anonymity requirement.
		
		Second, consider the case that~$k > \beta$ and thus $k'=\beta$.
		Let~$B_{D'}(t)$ be some block of~$D'$.
		We show that~$|B_{D'+S}(t)| = 0$ or~$|B_{D'+S}(t)| \ge k'$.
		If~$|B_{D'}(t)| \le 2s$, then~$B_{D'}(t) = B_{D}(t)$ (see \cref{line:NoOfVerticesCase2.1}).
		Hence, $B_{D'+S}(t) = B_{D+S}(t) = \emptyset$ since~$k > \beta > 4s$.
		If~$|B_{D'}(t)| > 2s$, then~$|B_{D}(t)| > k - 2s$ (see \cref{line:trivialNoInstance}) and thus~$|B_{D'}(t)| = \beta + \min\{2s,(|B_D(t)|-k)\}$ (see \cref{line:NoOfVerticesCase2.2}).
		Observe that~$|B_{D'+S}(t)| - |B_{D'}(t)| = |B_{D+S}(t)| - |B_{D}(t)|$, and thus,
		\begin{align}
			|B_{D'+S}(t)| 	& = (|B_{D+S}(t)| - |B_{D}(t)|) + |B_{D'}(t)|. \label{eq:size1}
		\end{align}
		Since~$|S|\le s$, we have~$|B_{D+S}(t)| - |B_{D}(t)| \ge -2s$. 
		We now distinguish the two cases~$|B_D(t)|-k \ge 2s$ and~$|B_D(t)|-k < 2s$.
		In the first case, it follows that~$|B_{D'}(t)| = \beta + 2s$ and from \cref{eq:size1} it follows
		\begin{align*}
			|B_{D'+S}(t)| & \ge -2s + \beta + 2s = \beta = k'.
		\end{align*}
		In the second case, it follows that~$|B_{D'}(t)| = \beta + |B_D(t)| - k$ (see \cref{line:NoOfVerticesCase2.2}).
		Observe that~$|B_{D+S}(t)| \ge k$ since~$|B_{D}(t)| > k - 2s$.
		From \cref{eq:size1} we conclude that
		\begin{align*}
			|B_{D'+S}(t)| & \ge k-|B_D(t)| + \beta + |B_D(t)|-k = \beta = k'.
		\end{align*}
                \hfill (Proof of \cref{lem:polyKernelHinrichtung-DirDegAnon})
	\end{proof}	
	The size of the kernel can be seen as follows:
	For each of the at most~$(\Delta_D+1)^2$ different blocks in the input graph~$D$, the algorithm keeps at most~$\beta + 2s = (\Delta_D+3)2s$ vertices in the set~$C$ (see \cref{line:NoOfVerticesCase1,line:NoOfVerticesCase2.1,line:NoOfVerticesCase2.2}).
	Thus, $|C| \in O( \Delta_D^3 s )$.
	The number of newly added vertices in \crefrange{line:add-outdeg-fix-vertices}{line:add-P-out-vertex} is at most~$\max\{\Delta_D^2\cdot |C|,k',\Delta_D+s+1\}$.
	Hence, $|P| \in O(\Delta_D^5 s)$ and thus the instance produced by \cref{alg:polyKernelDirDegAnon} contains at most~$O(\Delta_D^5 s)$ vertices.
	
	The running time can be seen as follows:
	Using bucket sort, one can lexicographically sort the~$n$ vertices by degree in~$O(n)$ time.
	Furthermore, in the same time one can create~$(\Delta_D+1)^2$ lists---each list containing the vertices of some degree~$t \in \N^2$.
	Then, the selection of the~$O(\Delta_D^3 s)$ vertices of~$C$ can be done in~$O(\Delta_D^3 s n)$ time.
	Clearly, inserting the vertices in~$P$ can be done in~$O(\Delta_D^5 s)$ time.
	Finally, inserting the arcs between the vertices in~$P$ (\cref{line:high-degree-in-P,line:separate-P-in-and-P-out}) takes~$O(\Delta_D^{10} s^2)$ time.\qed
\end{proof}

In contrast to both number problems in \cref{sec:DDCC,sec:DDSC}, we were unable to find a polynomial-time algorithm for the number problem for \DA, which is the special case of \nDDCSC asking for a $k$-anonymous target sequence.

\begin{problemdef}
    \problemtitle{\nDAlong (\nDA)}
    \probleminput{A sequence $\sigma = (c_1,d_1),\ldots,(c_n,d_n)$ of~$n$ nonnegative integer tuples,
                  two positive integers~$s$ and $k$.}
    \problemquestion{Is there a sequence $\sigma'=(c'_1,d'_1),\ldots,(c'_n,d'_n)$
      such that
      \begin{compactenum}[(i)]
        \item $\sum_{i=1}^n (c'_i-c_i) = \sum_{i=1}^n (d'_i-d_i) = s$,
        \item $c_i \le c'_i$, and $d_i \le d'_i$ for all $1 \le i \le n$, and
        \item each tuple in~$\sigma'$ appears at least~$k$ times?
      \end{compactenum}}
\end{problemdef}

We can show that \nDA is weakly NP-hard by a polynomial-time many-one reduction from \textsc{Partition}.

\begin{problemdef}
	\problemtitle{\textsc{Partition}}
	\probleminput{A multiset~$A=\{a_1,\ldots,a_n\}$ of positive integers that sum up to~$2B$.}
	\problemquestion{Is there a subset~$A'\subset A$ whose elements sum up to~$B$?}
\end{problemdef}

\begin{theorem}
	\label{thm:nDAhard}
	\nDA is (weakly) NP-hard even if~$k=2$.
\end{theorem}

\begin{proof}
	Given a multiset $A=\{a_1,\dots,a_n\}$, observe that we can assume without loss of generality that each integer in~$A$ is smaller than~$B$ (otherwise we could solve the instance in polynomial time).
 
	We create the following \nDA-instance with $s:=B$, $k:=2$, and the sequence $\sigma$ containing the following tuples.
	For each $a_i \in A$ create five tuples: one tuple $x_i$ of type $(2B(i+1)-a_i,0)$, one block~$X_i$ that contains two tuples of type~$(2B(i+1),0)$, and one block~$X'_i$ that contains two tuples of type~$(2B(i+1)-a_i,a_i)$.
	This completes the construction.
 
	We show that there is a subset $A'\subset A$ whose elements sum up to exactly~$B$ if and only if there is a sequence $\sigma'=(c_1',d_1'),\ldots,(c_n',d_n')$ that fulfills Conditions (i)--(iii) of \nDA.
 
	First, assume that there is some $A'\subset A$ and $\sum_{a\in A'} a=B$.
	Then, we obtain the desired sequence~$\sigma'$ by first copying~$\sigma$ and changing~$x_i$ as follows:
	For each $a_i \in A'$ change the tuple $x_i$ from type~$(2B(i+1)-a_i,0)$ to type~$(2B(i+1),0)$ and for each $a_i \notin A'$ change the tuple $x_i$ from type~$(2B(i+1)-a_i,0)$ to type~$(2B(i+1)-a_i,a_i)$.
	It is not hard to verify that this~$\sigma'$ is indeed a solution:
	For Condition~(i), observe that $\sum_{i=1}^n (c_i'-c_i)=s=B=\sum_{i=1}^n (d_i'-d_i)$ since the elements in~$A'$ as well as the elements in~$A \setminus A'$ sum up to~$B$.
	Condition (ii)~is clearly ensured by construction of~$\sigma'$.
	For Condition~(iii), note that in sequence~$\sigma'$ block~$X_i$ contains either $k=2$~tuples (if $a_i \notin A'$) or $k+1$~tuples (if $a_i \in A'$) and, analogously, note that block~$X'_i$ contains either $k=2$~tuples (if $a_i \in A'$) or $k+1$~tuples (if $a_i \notin A'$); $\sigma'$ contains no further tuples.
	
	Second, assume that there is a sequence $\sigma'=(c_1',d_1'),\ldots,(c_n',d_n')$ that is a solution for our constructed \nDA instance.
	First note that $\sigma'$ does not differ from~$\sigma$ ``a lot'' in the following sense.
	Since $s=B$ and $k=2$, in sequence~$\sigma'$ the first component and the second component of all tuples can in total be increased by at most~$B$, respectively.
	Next, observe that each tuple~$x_i$ must either be of type~$(2B(i+1),0)$ or of type~$(2B(i+1)-a_i,a_i)$, since  every other tuple is too far away (recall that $a < B$ for all~$a \in A$).
	This means that each tuple~$x_i$ contributes with~$a_i$ to the total sum over the differences in either the first component ($\sum_{i=1}^n (c_i'-c_i)$), or the second component ($\sum_{i=1}^n (d_i'-d_i)$).
	Since $\sum_{a \in A} a = 2B$, it follows that the tuples~$x_i$ require at least a budget of~$B$ in either the the first or the second component. 
	Let $A':=\{a_i \mid \text{$x_i$ is of type~$(2B(i+1),0)$ in~$\sigma'$}\}$.
	We show that $\sum_{a\in A'} a=B$.
	Assume towards a contradiction that $\sum_{a\in A'} a\neq B$.
	Since $\sum_{i=1}^n (c_i'-c_i)=\sum_{a\in A'} a$ and $\sum_{i=1}^n (d_i'-d_i)=\sum_{a\notin A'} a$, either $\sum_{i=1}^n (c_i'-c_i)$ or $\sum_{i=1}^n (d_i'-d_i)$ would be greater than~$B$---a contradiction to our budget.\qed
\end{proof}

Note that the hardness from \cref{thm:nDAhard} does not translate to instances of \nDA originating from digraph degree sequences because in such instances all numbers in the input sequence~$\sigma$ and also in the output sequence~$\sigma'$ are bounded by~$n-1$ where~$n$ is the number of tuples in~$\sigma$.
Since there are pseudo-polynomial-time algorithms for \textsc{Partition}, \cref{thm:nDAhard} leaves open whether \nDA is  strongly NP-hard or can be solved in polynomial time for instances originating from digraphs.

To again apply our framework (\cref{thm:FPT_transfer}), we show that \nDA is at least fixed-parameter tractable with respect to the largest possible integer~$\xi$ in the output sequence.
To this end, we develop an integer linear program that contains at most~$O(\xi^4)$~integer
variables and apply a famous result due to \citeauthor{Len83}~\cite{Len83}. %(see \cref{A:thm:nDA-FPT}).

\begin{theorem}\label{thm:nDA-FPT}
	\nDA is fixed-parameter tractable with respect to the largest possible integer~$\xi$ in the output sequence.
\end{theorem}

\begin{proof}
 Let $(\sigma,s,k)$ be an instance of \nDA.
 The key idea is that knowing how many tuples of type~$t$ in~$\sigma$ are
 transformed into a tuples of type~$t'$ in~$\sigma'$ for each pair~$\{t,t'\}$ of tuples
 is sufficient to describe a solution of our \nDA instance.
 To this end, observe that there are at most $(\xi+1)^2$ tuple blocks
 in~$\sigma$ and in~$\sigma'$, respectively.
 
 We describe an integer linear problem and create one variable $x_{t,t'}$ for each
 pair $t,t' \in \{0,\dots,\xi\}^2$ which denotes the number of tuples
 of type~$t$ in sequence~$\sigma$ that become tuples of type~$t'$ in sequence~$\sigma'$.
 We further use the binary variables~$u_{t}$ for each $t \in \{0,\dots,\xi\}^2$ being~$1$ if and only if some tuple of type~$t$ is used in the solution, that is,
 there is at least one tuple of type~$t$ in~$\sigma'$.
 We add a set of constraints ensuring that all tuples from $\sigma$ appear in $\sigma'$:
 $$\forall t \in \{0,\dots,\xi\}^2: \sum_{t' \in  \{0,\dots,\xi\}^2} x_{t,t'}=\lambda_\sigma(t).$$
 Then, we ensure that (i) holds by:
 $$\sum_{(t_1,t_2),(t_1',t_2') \in \{0,\dots,\xi\}^2} (t_1'-t_1) \cdot x_{(t_1,t_2),(t_1',t_2')} = s$$ 
 and by:
 $$\sum_{(t_1,t_2),(t_1',t_2') \in \{0,\dots,\xi\}^2} (t_2'-t_2) \cdot x_{(t_1,t_2),(t_1',t_2')} = s.$$
 We ensure that (ii) holds by:
 $$\forall(t_1,t_2),(t_1',t_2') \in \{0,\dots,\xi\}^2 \text{ with $t'_1<t_1$ or $t'_2<t_2$: } x_{(t_1,t_2),(t'_1,t'_2)}=0.$$
 We ensure that (iii) holds by:
 $$\forall t' \in \{0,\dots,\xi\}^2: \sum_{t \in  \{0,\dots,\xi\}^2} x_{t,t'} + k \cdot (1 - u_{t'}) \ge k.$$
 Finally, we add the following constraint set to ensure consistency between the~$u_t$ and $x_{t,t'}$~variables:\
 $$\forall t' \in \{0,\dots,\xi\}^2: \sum_{t \in  \{0,\dots,\xi\}^2} x_{t,t'} \leq u_{t'} \cdot n.$$

 Finally, fixed-parameter tractability follows by
 a result of \citet{Len83}
 % out-commented due to space
 (later improved by \citet{Kan87,FT87b})
 that says that an ILP with $\rho$ variables and $\ell$ input bits can be solved
 in~$O(\rho^{2.5\rho+o(\rho)}\ell)$ time.\qed
\end{proof}

Combining \cref{thm:PiEAfptDeltak,thm:FPT_transfer,thm:nDA-FPT} yields fixed-parameter tractability for \DA with respect to~$\degpara$.

\begin{corollary}\label[corollary]{cor:DDAfpt}
  \DA is fixed-parameter tractable with respect to~$\degpara$.
\end{corollary}

For undirected graphs, \citet{HNNS15} showed fixed-parameter tractability with respect to the maximum degree~$\Delta_G$ of the input graph.
This result was based on showing that~$\degpara \in O(\Delta_G^2)$.
For directed graphs, however, we can only show that~$\degpara \le 4 k (\Delta_D + 2)^2$ implying fixed-parameter tractability with respect to~$(k,\Delta_D)$.

\begin{lemma}\label[lemma]{lem:DA-kernel-k-delta-bound}
	Let~$D$ be a digraph and let~$S$ be a minimum size arc set such that~$D+S$ is $k$-anonymous.
	Then $\Delta_{D+S} \le 4 k (\Delta_D + 2)^2 + \Delta_{D}$.
\end{lemma}

\begin{proof}
 Let~$D = (V,A)$ be a digraph with maximum degree~$\Delta_D$ and let~$k$ be a positive integer.
 An arc set~$S \subseteq V^2$ is called \emph{$k$-insertion set} for~$D$ if~$D + S$ is~$k$-anonymous.
 Further, let~$S \subseteq V^2$ be a minimum size $k$-insertion set. % and suppose that~$|V(S)| \ge 4 k (\Delta_D + 2)^2$.
 We will show that if~$|V(S)| \ge 4 k (\Delta_D + 2)^2$, then the maximum degree in~$D+S$ is at most~$\Delta_{D}+2$, and
 if~$|V(S)| < 4 k (\Delta_D + 2)^2$, then the degree in~$D+S$ is clearly at most~$4 k (\Delta_D + 2)^2 + \Delta_{D}$.
	
Now suppose that~$|V(S)| \ge 4 k (\Delta_D + 2)^2$ and assume towards a contradiction that~$D+S$ has a maximum degree~$\Delta_{D+S} > \Delta_D+2$.
We next construct a smaller $k$-insertion set~$S'$ in two steps.
In the first step, we define for each vertex~$v \in V$, a target degree~$\tau(v) = (\tau^-(v),\tau^+(v))$ such that the following (and further conditions that are discussed later) holds:
	\begin{enumerate}[(a)]
		\item\label[cond]{cond:i} $\deg_D^-(v) \le \tau^-(v) \le \Delta_D + 2$,
		\item\label[cond]{cond:ii} $\deg_D^+(v) \le \tau^+(v) \le \Delta_D + 2$, and
		\item\label[cond]{cond:iii} the multiset~$\sigma(\tau) := \{\tau(v) \mid v \in V\}$ is $k$-anonymous, that is~$\lambda_{\sigma(\tau)}(\tau(v)) \ge k$ for each~$v \in V$.
	\end{enumerate}
	As a second step, we use \cref{lem:factor} to provide an arc set~$S'$ such that~$\sigma(D+S') = \sigma(\tau)$.
	Since~$\sigma(\tau)$ is $k$-anonymous, it follows that~$S'$ is a $k$-insertion set and we will show that~$|S'|<|S|$. 

	We now give a detailed description of the two steps and start with defining the target degree function~$\tau$ as follows
	\begin{align}
		\tau(v) := (\min\{\Delta_D + 1,\deg_{D+S}^-(v)\},\min\{\Delta_D + 1,\deg_{D+S}^+(v)\}). \label{eq:def-tau}
	\end{align}
	Observe that~$\tau$ satisfies the above three \cref{cond:i,cond:ii,cond:iii}. 
	Furthermore, we have~$\sum_{v \in V} \deg_{D+S}^+(v) > \sum_{v \in V} \tau^+(v)$ since the maximum degree in $D + S$ is larger than $\Delta_D + 2$.
	If we can realize the target degrees~$\tau$ with a~$k$-insertion set~$S'$, then it follows that~$|S'| < |S|$.
	
	To apply \cref{lem:factor} with~$\Delta_{D'} := \Delta_D + 2$, $x_i:=\tau^+(v_i) - \deg_D^-(v_i)$ and $y_i:=\tau^+(v_i) - \deg_D^+(v_i)$, we need to satisfy \cref{cond:maxdegbound,cond:indegbound,cond:outdegbound,cond:balance} of \cref{lem:factor}.
	By assumption, $\Delta_{D'} = \Delta_D + 2 < \Delta_{D+S} \le |V|-1$ holds.
	Hence, \cref{cond:maxdegbound} is fulfilled.
	Moreover, $\tau^-(v) \le \Delta_{D'}$ and~$\tau^+(v) \le \Delta_{D'}$ holds for all $v\in V$.
	\cref{cond:indegbound,cond:outdegbound} are thus also satisfied.
	However, we also need to ensure~$\sum_{i=1}^nx_i=\sum_{i=1}^ny_i$ (\cref{cond:balance}), that is, we need to ensure that~$\tau$ changes the indegrees and outdegrees of the vertices in~$D$ by the same overall amount.
	This might not be true as we changed the indegrees and outdegrees independently.
	To overcome this problem, we subsequently adjust~$\tau$ again.
	
	Assume without loss of generality that compared to~$D+S$ the target degree function~$\tau$ reduced more indegrees than outdegrees, that is,
        $$\sum_{v \in V} (\tau^+(v) - \deg_D^+(v)) > \sum_{v \in V} (\tau^-(v) - \deg_D^-(v)).$$
	Denote by $\diff_\tau$ the difference between the two sums, that is,
	\begin{align*}
		 \diff_\tau := & \sum_{v \in V} \left((\tau^+(v) - \deg_D^+(v)) - (\tau^-(v) - \deg_D^-(v))\right) \\
                 = & \sum_{v \in V} \left(\tau^+(v)- \tau^-(v) + \deg_D^-(v) - \deg_D^+(v)\right)\\
                 = & \sum_{v \in V} \left(\tau^+(v)- \tau^-(v)\right) + \sum_{v \in V} \left(\deg_D^-(v) - \deg_D^+(v)\right)\\
                 = & \sum_{v \in V} \left(\tau^+(v)- \tau^-(v)\right).
	\end{align*}
	Further, denote by~$B_\tau(\tau(v))$ the \emph{block of~$v$ in~$\tau$}, that is the set vertices having the same target degree as~$v$.
	In the final adjustment of~$\tau$ we need~$\diff_\tau$ to be at least~$k$ and at most~$3k$.
	Hence, if~$\diff_\tau < k$, then we adjust~$\tau$ as follows:
	Pick an arbitrary vertex~$v$ such that the outdegree of~$v$ in~$D+S$ is larger than~$\Delta_D + 1$.
	Observe that such a vertex must exist: 
	We assumed to reduce the indegrees more than the outdegrees (thus $0 < \diff_\tau$), hence we reduced the indegrees of the vertices of at least one block, that is, of at least~$k$ vertices. 
	Since~$\diff_\tau < k$ it follows that we also reduced the outdegrees of at least one block and thus, such a vertex~$v$ exists.
	If the block of~$v$ contains at least~$2k$ vertices, then increase the target outdegree of exactly~$k$ of these vertices by one.
	Otherwise, if the block contains less than~$2k$ vertices, then increase the target outdegree of all these vertices by one.
	It follows that~$\diff_\tau > k$.
	Furthermore, observe that~$\sum_{v\in V} \tau^+(v) < \sum_{v \in V} \deg_{D+S}^+(v)$, that is, after realizing the target degrees~$\tau$, the corresponding $k$-insertion set~$S'$ is still smaller than~$S$.

	In the following, we increase the indegrees in two rounds.
	Observe that if we do not increase outdegrees, then it still holds that~$|S'|<|S|$.
	In the first round, while~$\diff_\tau \ge 3k$ do the following:
	\begin{enumerate}
		\item Pick an arbitrary vertex with~$\tau^-(v) \le \Delta_D$.
		\item\label{2nd} If~$|B_\tau(\tau(v))| \le 2k$, then increase the target indegree~$\tau^-(u)$ by one for each~$u \in B_\tau(\tau(v))$.
		\item\label{3rd} Else, it follows that~$|B_\tau(\tau(v))| > 2k$. 
				Let~$B' \subseteq B_\tau(\tau(v))$ be an arbitrary subset of size exactly~$k$ and increase the target indegree~$\tau^-(u)$ by one for each~$u \in B'$.
	\end{enumerate}
	Observe that in Step~\ref{2nd} as well as in Step~\ref{3rd} we increase the target indegree of at least~$k$ vertices that have the same target degree.
	Furthermore, in Step~\ref{3rd} we ensure that at least~$k$ vertices with the original target degree remain.
	Hence, the (changed) multiset~$\sigma(\tau)$ is still $k$-anonymous.
	Furthermore, it is easy to verify that the maximum target indegree is at most~$\Delta_D + 1$.
	Finally, observe that we decrease~$\diff_\tau$ in each iteration by at most~$2k$ and, hence, we have $\diff_\tau \ge k$.
	
	In the second round, we have that~$k \le \diff_\tau < 3k$.
	We simply pick a block~$B_\tau(\tau(v))$ with at least~$4k$ vertices and increase the target indegree of exactly~$\diff_\tau$ vertices.
	Since~$|V(S)| \ge 4 k (\Delta_D + 2)^2$ and there are at most~$(\Delta_D + 2)^2$ different degrees in~$\tau$ (in- and outdegrees between~$0$ and~$\Delta_D+1$), it follows that there exists such a block of size at least~$4k$.
	Furthermore, observe that, after this change in the second round,~$\sigma(\tau)$ is still $k$-anonymous and the maximum target indegree is at most~$\Delta_D + 2$.
	Hence, the adjusted target degree function~$\tau$ fulfills \cref{cond:maxdegbound,cond:indegbound,cond:outdegbound,cond:balance} of \cref{lem:factor}.
        
	It remains to show the last condition in \cref{lem:factor}, that is, \cref{cond:minsize} stating $s = \sum_{i=1}^nx_i \ge 2\Delta_{D'}^2 + \Delta_{D'}$.
	Due to the definition of~$\tau$ (see \cref{eq:def-tau}), it follows that we only decreased the degrees of vertices with in- or outdegree greater than~$\Delta_D+1$ in~$D+S$.
	Since the target degrees of these vertices is at least~$\Delta_D+1$ (the later changes to~$\tau$ only increased some degrees), it follows that~$V(S)$ is exactly the set of vertices whose target indegree (outdegree) is larger than their indegree (outdegree) in~$D$. 
	Hence, 
	\begin{align*}
		\sum_{v \in V} \tau^+(v)-\deg_D^+(v) &\ge |\{v \in V \mid \tau^+(v) > \deg_D^+(v)\}|\\
		&= |V(S)| \ge 4 k (\Delta_D + 2)^2.
	\end{align*}
	Since~$\Delta_{D'} = \Delta_D+2$ it follows that \cref{cond:minsize} is indeed fulfilled.
	Thus, the set~$S' := A'$ realizing~$\tau$ is a $k$-insertion set of size less than~$|S|$; a contradiction to the fact that~$S$ is a minimum size $k$-insertion set for~$D$.\qed
\end{proof}

Combining \cref{thm:PiEAfptDeltak,thm:FPT_transfer,thm:nDA-FPT,lem:DA-kernel-k-delta-bound}, we obtain the following.

\begin{corollary}\label[corollary]{cor:DA-FPT-deltastern}
	\DA is fixed-parameter tractable with respect to~$(k,\Delta_D)$.
\end{corollary}

It remains open whether \DA is fixed-parameter tractable with respect to~$\Delta_D$.
We remark that the problems \DDCC and \DDSC are both NP-hard for~$\Delta_D = 3$. 
This follows from an adaption of the construction given by \citet[Theorem 3.2]{Mil15}.

\section{Conclusion}
We proposed a general framework for digraph degree sequence completion
problems and demonstrated its wider applicability in case studies. 
Somewhat surprisingly, the presumably more technical case 
of digraphs allowed for some elegant tricks (based on
flow computations) that seem not to work for the 
presumably simpler undirected case. 
Once having established the framework (see Section~\ref{sec:GeneralSetting}), 
the challenges then
associated with deriving fixed-parameter tractability and kernelizability 
results usually boil down to the question for fixed-parameter tractability and (pseudo-)polynomial-time solvability of a simpler
problem-specific number problem. 
While in most cases we could develop polynomial-time algorithms solving these number problems, in the case of \DAlong the polynomial-time solvability of the associated number problem remains open.
Moreover, a widely open field is to attack weighted versions 
of our problems. 
Finally, we believe that due to the fact that many real-world networks 
are inherently directed 
(e.g., representing relations such as ``follower'', ``likes'', or ``cites'')
further studies (e.g., exploiting special digraph properties) 
of digraph degree sequence completion problems are desirable.

\bibliographystyle{abbrvnat}
\bibliography{bibfile}

% \appendix
% \newpage
% \section{Proofs}
% \markboth{Appendix\hfill Digraph Degree Sequence Completion Problems}{Appendix\hfill Digraph Degree Sequence Completion Problems}

% \appendixProofText

\end{document}